\documentclass[a4paper,11pt]{article}

\usepackage{authblk}

\usepackage{amsmath,bm}
\usepackage{amsthm}
\usepackage{amssymb}

\usepackage{etoolbox}
\usepackage{mathtools}

\usepackage{algorithmic}
\usepackage{algorithm}

\usepackage{tikz}
\usetikzlibrary{shapes.geometric,arrows,decorations.markings}

\usepackage{xcolor}

\usepackage{colortbl}
\usepackage{array,arydshln,multirow}

\usepackage{paralist}

\usepackage{graphicx}

\usepackage[left=1.0in,top=1.0in,right=1.0in,bottom=1.0in,nohead]{geometry}

\usepackage{wrapfig}

\usepackage{cite}

\theoremstyle{definition}
\newtheorem{definition}{Definition}

\newtheorem{theorem}{Theorem}

\newtheorem{lemma}[theorem]{Lemma}

\newtheorem{prop}[theorem]{Proposition}

\newtheorem*{prop*}{Proposition}

\patchcmd{\subequations}{}%
{}{}{}

\graphicspath{{figs/},{new_figs/}}

\makeatletter
\newcommand{\leqnomode}{\tagsleft@true}
\newcommand{\reqnomode}{\tagsleft@false}
\makeatother

\newcommand{\CEIL}[1]{\left\lceil #1\right\rceil}
\newcommand{\BIGP}[1]{\left( #1\right)}
\newcommand{\BIGBP}[1]{\left\{ #1\right\}}

\newcommand{\BIGC}[1]{\left| #1\right|}

\newcommand{\OP}[1]{\operatorname{#1}}
\newcommand{\MC}[1]{{\mathcal #1}}
\newcommand{\MBB}[1]{{\mathbb #1}}

\long\def\longdelete#1{}

\title{\bf Iterative Partial Rounding for Vertex Cover \\ with Hard Capacities}

\author[ ]{Mong-Jen Kao}

\affil[ ]{Institute of Information Science,}
\affil[ ]{Academia Sinica, Taiwan.}
\affil[ ]{\small mong@iis.sinica.edu.tw}

\date{}
\begin{document}

\begin{titlepage}

\maketitle
\thispagestyle{empty}

\begin{abstract}
We provide a simple and novel algorithmic design technique, for which we call \emph{iterative partial rounding}, that gives a tight rounding-based approximation for vertex cover with hard capacities (VC-HC).
% and its variations.
%
In particular, 
%
%\begin{itemize}
%	\item
		we obtain an $f$-approximation for VC-HC on hypergraphs, improving over a previous results of 
		%$2.155$ and $2f$ due to 
		Cheung et al. (SODA 2014) to the tight extent.
		%, and 
		This also closes the gap of approximation 
		%for this problem since 
		%since this problem
		since it was posted by Chuzhoy and Naor in (FOCS 2002).
%		
%	\item
%		For Partial VC-HC, an $(f+\epsilon)$-approximation is obtained, improving over a previous ratio of $(2f+2)(1+\epsilon)$.
		%-approximation by Cheung et al. (SODA 2014).
%		
%	\item
%		For (weighted) vertex cover with soft capacities, we show that our rounding approach gives a tight result which essentially matches earlier works for this problem.
		%
%\end{itemize}
%
We believe that our rounding technique is of independent interest when hard constraints are considered.

\smallskip

Our main technical tool for establishing the approximation guarantee is a separation lemma that certifies the existence of a strong partition for solutions that are basic feasible in an extended version of the 
natural LP.
\end{abstract}

\bigskip

\begin{minipage}{0.95\textwidth}
\textbf{Keywords:} iterative partial rounding, capacitated vertex cover, hard capacities, approximation algorithm
\end{minipage}

\end{titlepage}

\section{Introduction}

We consider the vertex cover problem with hard capacity constraints (VC-HC) on hypergraphs.
In this problem, we are given a hypergraph $G=(V,E)$ with a maximum edge size $f$.
%, where e
Each $e \in E$ is associated with an edge demand $d_e$ and each $v \in V$ is associated with a capacity $c_v$ and an available multiplicity (the number of available copies) $m_v$.
The objective is to find a minimum-size vertex cover, a vertex multi-set represented by an assignment function, such that the demands of the edges can be covered by the capacities of the vertices chosen in the multiset and the multiplicity of each vertex does not exceed its available multiplicity.
When no upper bound is imposed on the multiplicities of each vertex, i.e., $m(v) = \infty$ for all $v$, this problem is then referred to as (soft) capacitated vertex cover (CVC).
%, for which 
In this case weighted vertex set is usually considered and minimum-weight vertex multi-set is sought.
In this paper, we assume that VC-HC takes unweighted vertices unless specified otherwise.

\smallskip

The study of VC-HC was initiated in the notable work of Chuzhoy and Naor~\cite{1151271}, where normal graphs with unit edge demand are assumed, i.e., $f=2$ and $d_e=1$ for all $e$.
Under this setting, they established a surprising result that, while 
%VC-HC
this setting admits constant factor approximations, it becomes set-cover hard when $\{0,1\}$-weighted vertices are considered, i.e., the weight of each vertex can be either $0$ or $1$.
This implies an interesting logarithmic separation on the approximability %of VC-HC 
between weighted and unweighted vertices.
%\mong{This somehow contrast *** blahblah *** footnote for capacitated problems later}
%
%\smallskip
%
%
In the same work, the status of (unweighted) VC-HC with general edge demand
%, i.e., graphs with parallel edges, 
was left as an open problem.
While the gap of approximation for this problem was settled partially a decade later by Saha and Khuller~\cite{Saha:2012:SCR:2359332.2359443} and Cheung~et~al~\cite{doi:10.1137/1.9781611973402.124},
% roughly a decade later,
% in that constant factor approximations are obtained, 
the exact approximability of this problem remains unsettled.

\smallskip

For a brief introduction on the current status of VC-HC,
Chuzhoy and Naor~\cite{1151271} presented a $3$-approximation for normal graphs with unit edge demand.
This result was later improved by Gandhi~et~al~\cite{Gandhi:2006:IAA:1740416.1740428} to a tight $2$-approximation with a refined approach.
Saha and Khuller~\cite{Saha:2012:SCR:2359332.2359443} considered general edge demands and 
%.
%They presented an $O(1)$-approximation for normal graphs and 
presented an $O(f)$-approximation for hypergraphs.
This result was improved by Cheung~et~al~\cite{doi:10.1137/1.9781611973402.124} with a $\BIGP{1+2/\sqrt{3}} \approx 2.155$-approximation for normal graphs and a $2f$-approximation for hypergraphs.

\smallskip

One intriguing thing in the development of VC-HC is how the techniques that were used to solve this problem are influenced (constrained)
%, or more precisely, constrained, 
by the complexity separation between weighted and unweighted vertex sets.
%
%While VC-HC is as hard as set cover for weighted vertex sets, it actually admits constant factor approximations for unweighted vertex sets.
%
Therefore, if an approach were to work, it has to be sensitive enough to tell the difference between the assumptions.
This nature, as also pointed out in~\cite{doi:10.1137/1.9781611973402.124}, renders existing techniques for CVC, such as primal-dual~\cite{989540}, dependent rounding~\cite{Gandhi200455}, LP rounding, etc., not directly applicable to VC-HC since very often they are not sensitive to the weight of the vertices.

\smallskip

In fact, {all} of existing results for VC-HC are built on the same two-staged rounding principle:
First, it begins with a vertex-side threshold rounding.
Then one or multiple edge-vertex patching procedure is introduced to meet the covering guarantee.
The main challenge of this approach has been on devising a delicate patching procedure.
This has been demonstrated in the 
%transition from the work of Chuzhoy and Naor~\cite{1151271} to the work of Gandhi~et~al~\cite{Gandhi:2006:IAA:1740416.1740428} 
%works~\cite{1151271,Gandhi:2006:IAA:1740416.1740428} for the case of normal graphs with unit edge demand, and 
%the transition from the work of Saha and Khuller~\cite{Saha:2012:SCR:2359332.2359443} to the work of Cheung~et~al~\cite{doi:10.1137/1.9781611973402.124} 
the works~\cite{Saha:2012:SCR:2359332.2359443,doi:10.1137/1.9781611973402.124}.
% for hypergraphs with general edge demands.
%
In particular, a near-tight $2.155$-approximation for normal graphs is obtained, using an elegant matching structure extracted during the patching stage, combined with a neat interaction back to the threshold parameter of the first stage.
Although improved approximations are obtained,
%However, 
it seems that current two-staged rounding techniques have reached their limitations,
%unlikely that further improvement could be extracted in this rounding-patching scenario, 
and significant new ideas are required to close the gap.

\medskip

%
%From the discussions above,
Therefore, a natural and central question that arises is thus:
\begin{quote}
\emph{Can the rounding for VC-HC be done without patching?}
\end{quote}
In this work, we provide a positive answer to the above question:
We present a novel rounding-based approximation algorithm for VC-HC which closes the gap of approximation since it was posted a decade ago.
Compared to prior results, our algorithm is very simple to state and easy to understand.
We believe that our rounding technique is also of independent interest when hard constraints are considered.
%

%

%%%

%

\subsection{Our Results}
%\paragraph{Our Contributions.}

%
Our approximation algorithm is stated as follows.

\begin{theorem} \label{thm-informal-tight-vc-hc}
There is a rounding-based approximation algorithm for VC-HC that, given an instance with largest size of hyperedges $f \ge 2$, produces an $f$-approximation in polynomial-time.
%, where $f \ge 2$ is the size of the largest hyperedge.
\end{theorem}

%

%\noindent
%
%This improves over the previous result of Cheung et al~\cite{doi:10.1137/1.9781611973402.124} to a tight extent and closes the gap of approximation since this problem was posted by Chuzhoy and Naor in~\cite{1151271}.

%\smallskip

%
Our algorithm is driven by a simple \emph{iterative partial rounding} scheme.
In each iteration, it makes partial decisions based on current working LP and rounds the vertices \emph{fractionally}.
When no such decisions are there to be made, it rounds up all vertices \emph{unconditionally} and stops.

\smallskip

During this process, the input instance along with the working LP are modified gradually.
%, and t
While the edges may be removed (folded) from the instance as the algorithm iterates, the lower bounds on the multiplicity of the vertices are meanwhile strengthened.
The intermediate solutions to the working LPs may not be feasible for the original input instance.
However, we guarantee that the overall optimality is not lost before the final rounding is done, and all the partial decisions we made entail no potential loss in the feasibility and approximation guarantee of the final solution.
Together this gives our tight approximation for this problem.

\smallskip

In contrast to previous two-staged approaches, instead of making one big decisive rounding followed with patching, our algorithm makes only partial indecisive moves which are also proven harmless and waits before the final decisive rounding can be made.
From this point of view, we also believe that our rounding scheme is of independent interest when other problems with hard constraints are considered.

\smallskip

Our technical tool for establishing the approximation guarantee for the final rounding step is a separation lemma that certifies the conditional existence of a strong partition on the vertices, given by solutions that are basic feasible for the LPs our algorithm is working with.
%
%
%\smallskip
%
%
The presented separation lemma can be seen as a recast of the well-established Carath\'{e}odory's theorem~\cite{Schrijver:1986:TLI:17634} in the language of VC-HC polytope.
Intuitively, it says that: A vertex-to-edge correspondence with strong separation guarantee exists at places where 
%the the VC constraint
a certain type of linear constraint is inactive.
Similar ideas have been used, e.g.,\cite{doi:10.1137/1.9781611973402.124}, in the usage of the multi-set multi-cover polytope~\cite{Saha:2012:SCR:2359332.2359443}, which by formulation is a vertex-to-vertex matching structure, to extract a vertex-to-vertex matching.
When the VC-HC polytope is considered, however, we have an essentially different entangling relation to deal with, in particular, relation between the vertices $V$, the edges $E$, and their cross-products $E\times V$. 
Our separation lemma provides an interpretation from the perspective of Carath\'{e}odory's theorem to this situation:
By properly removing the entangling $E \times V$ relation,
% while assuming a reasonable entanglement outside $I$, 
a strong and meaningful structural guarantee can be extracted for our rounding approach.

\smallskip

The rest of this paper is organized as follows.
In Section~\ref{Preliminary} we formally define VC-HC and the extended LP relaxation we will be using throughout this paper.
In Section~\ref{sec-VC-constraint} we introduce the key notion of edge-folding that motivates our iterative rounding approach.
We present our tight approximation algorithm in Section~\ref{sec-iterative-partial-rounding-VC-HC} and establish the approximation guarantee in Section~\ref{sec-approx-ratio}.
%

%

%%
%%

%%%
%%
%%

\subsection{Further Related Work}
\label{sec-related-work}
%\paragraph{Related Work.}
%

In the following we briefly summarize other related results for vertex cover (VC), CVC, and VC-HC with relaxed constraints.

%
%
%\paragraph{Vertex cover.}

%
%For partial covers, 
%allowing an $O(\epsilon)$-fraction of loss in the number of edges, 
%an~$\BIGP{\MC{O}(\log\frac{1}{\epsilon}), 1-\epsilon}$-covering
%a cover with cost ratio of $O(\log\frac{1}{\epsilon})$ to the optimal full cover 
%can be obtained via simple randomized rounding~\cite{Vazirani:2001:AA:500776}.
%

%
%In the classical set cover problem, we are given a universe $\MC{E}$ of elements and a collection $\MC{S}$ of sets defined over $\MC{E}$. The objective is to find a collection $\MC{C} \subseteq \MC{S}$ with the minimum cardinality such that the elements of $\MC{E}$ are covered by $\MC{C}$, i.e., $\bigcup_{C\in \MC{C}}C$ equals $\MC{E}$. 
%
%The set cover problem is a classic NP-hard problem, extensively studied in the literature, and the best approximation factors achievable are (1) $\log n$, where $n$ is the size of the universe~\cite{VC79,285059}, and (2) $\Delta_c$, where $\Delta_c = \max_{e\in \MC{E}}\BIGC{\BIGBP{C \in \MC{S}\colon e\in C}}$ is the largest frequency of the elements~\cite{Vazirani:2001:AA:500776}.
%
%When the cardinality of each set is bounded by a constant $k$, i.e, $\BIGC{S} \le k$ for all $S \in \MC{S}$, $O(\log k)$-approximations are known for this problem~\cite{Duh:1997:AKS:258533.258599}.
%

\vspace{-8pt}

\paragraph{Soft capacitated covering.}

For vertex cover, it is known that a $f$-approximation can be obtained by both LP rounding and LP duality~\cite{BarYehuda1981198,doi:10.1137/0211045} for hypergraphs.
%, which also extend to an $f$-approximation for hypergraphs.
% with largest edge size $f$.
%
Khot and Regev~\cite{Khot2008335} showed that, by assuming the unique game conjecture (UGC), approximating this problem to a ratio better than $f - \epsilon$ is NP-hard for any $\epsilon > 0$ and $f \ge 2$.

\smallskip

%
%The capacitated vertex cover generalizes vertex cover in that a demand-to-service assignment model is evolved from the original 0/1 covering model.
%
The soft capacitated vertex cover problem was first introduced in the notable work of Guha~et~al~\cite{989540},
in which a $2$-approximation was presented using primal-dual approach.
This result extends to $f$-approximation for hypergraphs.
Gandhi et al.~\cite{1181955} provided another $2$-approximation via dependent rounding.
% for bipartite graphs.
%
%Bar-Yehuda et al.~\cite{doi:10.1137/080728044} considered partial VC-SC and presented a $3$-approximation for simple graphs that extends to $(f+1)$-approximations for hypergraphs, based on local ratio techniques.
%
%
%\smallskip
%
%
Kao~et~al~\cite{MJKHLCDTL13,springerlink:10.1007/s00453-009-9336-x,MJKaoDissertation12} considered capacitated dominating set, an alternative notion of capacitated covering, and presented a series of results studying the complexity and approximability of this problem when different classes of graphs are considered.
%studied the complexity and approximability of this problem in a series of results.
%
Special cases and variations of this problem were also considered independently~\cite{springerlink:10.1007/978-3-642-13731-08,DBLP:conf/iwpec/DomLSV08,Liedloff:2010:SCD:1939238.1939250}.
%

%
%\paragraph{Submodular cover.}
%
%Wolsey~\cite{Wolsey1982} considered submodular set cover, which includes classical set cover as a special case and which relates to capacitated covering in a simplified form, and presented a $(\ln \max_S f(S) +1)$-approximation.
%
%This approach was generalized by Chuzhoy and Naor~\cite{1151271} to capacitated set cover with hard capacities and unit demands, for which a $(\ln \delta+1)$-approximation was presented, where $\delta$ is the maximum size of the sets.
%
%This approach does not, however, give the same guarantee for general demands as the pseudo-polynomial factor pops up in the approximation ratio.
%

\vspace{-8pt}

\paragraph{VC-HC with relaxed constraints.}
For partial VC-HC, which aims at covering a given number of edges, Cheung~et~al~\cite{doi:10.1137/1.9781611973402.124} provided a $(1+\epsilon)(2f+2)$-approximation, based on a reduction 
%from $f$-hypergraphs 
to $(f+1)$-hypergraphs and their $2f$-approximation for VC-HC presented in the same work.

\smallskip

Gandhi~et~al~\cite{Gandhi:2006:DRA:1147954.1147956} considered weighted VC-HC with relaxed multiplicity constraints and presented an augmented $(2,2)$-covering\footnote{An augmented $(\alpha,\beta)$-covering is a solution that violates the multiplicity limit by a factor $\alpha$ and has a cost factor $\beta$ to the optimal LP solution for VC-HC.} for normal graphs.
Grandoni et al.~\cite{DBLP:journals/siamcomp/GrandoniKPS08} considered unit vertex multiplicity, i.e., $m(v)=1$, and presented an augmented $(f,f^2)$-covering for hypergraphs.
%This approach generalizes to $f$-hypergraphs: By allowing at most $f$ multiplicities for each vertex, a cover with a cost ratio of $f^2$ can be obtained.
%
These results were improved by Kao et al.~\cite{KaoTuLee15} to augmented $\vphantom{ {\dfrac{}{}}^{\frac{T}{T}}} \big( s, (1+\frac{1}{s-1})(f-1)\big)$-covering for any $s \ge 2$, using an extended dual-fitting approach.

\subsection{Other hard-capacitated problems.}

In recent years, significant progress have been made on hard-capacitated problems, including VC-HC, Capacitated facility location (CFL), Capacitated $k$-center, and $k$-median.
We summarize the current progress and recent breakthroughs for the respective problems in Table~\ref{table-summary-capacitated-problems}, using a condensed format.
We also refer the reader to the references therein for further details.

\smallskip

{
\begin{table*}[htp]
	\renewcommand{\arraystretch}{1.2}
	\setlength{\dashlinegap}{2pt}
	\setlength{\dashlinedash}{2pt}
	\newcolumntype{C}[1]{>{\centering\arraybackslash}m{#1}}
	\newcolumntype{\Empty}{@{}m{0pt}@{}}
	\begin{tabular}{ | c | @{} C{4.2em} @{} | @{} C{4.2em} @{} | @{} C{10em} @{} | c \Empty |}
		\hline
		%%%
		%	VC
		%
		\multirow{5}*[2pt]{ VC }
		&
		\multirow{5}*[-4pt]{
			\begin{tabular}{c}
				open \\ cost
			\end{tabular}
		}
		&
		CVC
		&
		\multicolumn{2}{l}{ \enskip $f$-approx via primal-dual~\cite{989540} }
		& \\[8pt]\cline{3-4}\cdashline{5-5}
		%%
		%	VC-HC
		& &
		\multirow{4}{*}{ VC-HC }
		&
		weighted & unweighted
		& \\\cdashline{4-5}
		& & &
		set-cover-hard~\cite{1151271} 
		& 
		\begin{tabular}{c}
			$f$-approx via \\[-3pt]
			\emph{iterative partial rounding}
		\end{tabular}
		& \\[20pt]\cdashline{4-4}\cline{5-5}
		%%
		%	set cover
		%
		\multirow{1}*[2pt]{( Set Cover )} & & &
		\multicolumn{2}{l}{ \enskip logarithmic approx.~\cite{Wolsey1982,1151271} }
		& \\[8pt]\hline
		%
		%%%
		%
		%	FL
		%
		\multirow{3}*{ 
			\begin{tabular}{c}
				Facility \\ Location
			\end{tabular}
		}
		&
		\multirow{3}*[-4pt]{
			\begin{tabular}{c}
				open \\
				+ \\
				assign
			\end{tabular}
		}
		&
		UFL
		&
		\multicolumn{2}{l}{
		\begin{tabular}{l}
			$1.488$-approx~\cite{Li:2011:AAU:2027223.2027230}, \\
			NP-hard to approx within $(1.463-\epsilon)$~~\cite{Guha:1998:GSB:314613.315037}
		\end{tabular}
		}
		& \\[20pt]\cline{3-5}
		& &
		CFL
		&
		\multicolumn{2}{c}{
		\begin{tabular}{l \Empty}
			$5$-approx via local search~\cite{Bansal:2012:CFL:2404160.2404173}, & \\[0pt]
			$O(1)$-approx, LP-rounding via MFN-relaxation~\cite{DBLP:conf/focs/AnSS14}
		\end{tabular}
		}
		& \\[24pt]\hline
		%%%
		%
		%	k-center
		%
		%
		\multirow{2}*[-2pt]{ $k$-center }
		&
		\multirow{4}*[-8pt]{ 
			\begin{tabular}{c}
				assign \\
				cost
			\end{tabular}
		}
		&
		\multirow{2}*[-4pt]{ $\max \min$ }
		&
		\multicolumn{2}{l}{ \quad	$2$-approximation~\cite{Hochbaum:1985:BPH:2775965.2775967} }
		& \\[6pt]\cline{4-5}
		%%
		%	C-k-center
		%
		& & &
		\multicolumn{1}{c:}{ Hard-capd. }
		&
		$9$-approximation~\cite{DBLP:journals/mp/AnBCGMS15}
		& \\[6pt]\cline{1-1}\cline{3-3}\cline{4-5}
		%%
		%
		%	k-median
		%
		\multirow{3}*[-2pt]{ $k$-median }
		& &
		\multirow{3}*{ avg }
		&
		\multicolumn{2}{c}{ 
			\begin{tabular}{l}
			uncapacitated, \\
			$(2.675+\epsilon)$-approx via pseudo-approx~\cite{Byrka:2015:IAK:2722129.2722179,Li:2013:AKV:2488608.2488723}
			\end{tabular} }
		& \\[10pt]\cline{4-5}
		%
		%	capacitated
		%
		& & &
		\multicolumn{2}{c}{ 
			\begin{tabular}{l}
				uniform hard-capacitated, \\
				$e^{O(1/\epsilon^2)}$-bi-approx, using $(1+\epsilon)\cdot k$ facilities~\cite{Li:2015:UCK:2722129.2722176}
			\end{tabular} }
%		&
%		\begin{tabular}{l}
%			$e^{O(1/\epsilon^2)}$-bi-approx, using \\[-2pt]
%			$(1+\epsilon)\cdot k$ facilities~\cite{Li:2015:UCK:2722129.2722176}
%		\end{tabular}
		%
		& \\[24pt] \hline
	\end{tabular}
%	\vspace{-2pt}
	\caption{A summary on related hard-constrained problems.}
	\label{table-summary-capacitated-problems}
\end{table*}
}

\section{Problem Statement and LP Relaxation}
% and Separation Lemma} 
\label{Preliminary}

In this section we define VC-HC formally and introduce the extended LP relaxation we will be using in this paper.

\smallskip

%\paragraph{Notations.}
We begin with basic graph notations.
Throughout this paper, we use $G=(V,E)$ to denote a hypergraph with vertex set $V$ and edge set $E \subseteq 2^V$. Note that, each hyperedge $e \in E$ is represented by the set of its incident vertices, which is a vertex subset of $V$.
We use $f := \max_{e \in E}|e|$ to denote the size of the largest hyperedge in the considered graph $G$.

\smallskip

For any edge subset $\MC{E} \subseteq E$, we use $\MC{E}[v]$ to denote the set of edges in $\MC{E}$ that are incident to the vertex $v$.
Formally, $\vphantom{{\dfrac{}{}}^{\frac{T}{T}}} \MC{E}[v] := \big\{\hspace{2pt} e \hspace{2pt} \colon \hspace{2pt} e\in \MC{E} \text{ such that } v \in e \hspace{2pt}\big\}$.
This definition extends to sets of vertices, i.e., for any $U \subseteq V$, $\MC{E}[U] := \bigcup_{v \in U}\MC{E}[v] = \big\{\hspace{2pt} e \hspace{2pt} \colon \hspace{2pt} e\in \MC{E} \text{ such that } e \cap U \neq \emptyset \hspace{2pt}\big\}$.

\paragraph{Minimum Vertex Cover with Hard Capacities (VC-HC).}
In this problem we are given a hypergraph $G = (V,E)$, where each $e\in E$ is associated with a demand $d_e$ and each $v \in V$ is associated with a capacity $c_v$ and an (integral) available multiplicities $m_v$.

\smallskip

A solution to this problem consists of an assignment $h\colon E \times V \rightarrow {\mathbb{R}}^+ \cup \{0\}$, where $h_{e,v}$ denotes the fraction of the edge $e$ that is assigned to the vertex $v$.
The multiplicity given by $h$ is denoted $\vphantom{{\dfrac{}{}}^{\dfrac{}{T}}} x^{(h)}_v := \Big\lceil \sum_{e \in E[v]} (d_e\cdot h_{e,v}) / {c_v} \Big\rceil$.
The assignment $h$ is said to be \emph{feasible} if $\sum_{v \in e}h_{e,v} =1$ for all $e\in E$ and $x^{(h)}_v \le m_v$ for all $v \in V$.
%
%The \emph{weight (cost)} of $h$, denoted $w(h)$, is defined to be $\sum_{v \in V} w(v)\cdot x_h(v)$.
%

\medskip

Given an instance $\Pi = (V, E, \mathbf{c}, \mathbf{m}, \mathbf{d})$ as described above, the problem VC-HC is to find a feasible assignment $h$ such that $\vphantom{{\dfrac{}{}}^{\frac{T}{T}}} \sum_{v \in V}x^{(h)}_v$ is minimized.
Without loss of generality, we assume that the input graph $G$ admits a feasible assignment since this condition can be checked via a max-flow computation.
% including all the available multiplicities.
%, via a max-flow computation.
%

\smallskip

We also remark that, when $d_e$ and $c_v$ are integer-valued for all $e \in E$, $v \in V$, by the integrality of b-matching polytope, any fractional assignment can be turned into an integral assignment, i.e., $d_e\cdot h_{e,v} \in \MBB{Z}^{\ge 0}$ for all $e,v$, using a standard integer flow computation.
% without changing the multiplicity of the cover.
%

\paragraph{Extended LP relaxation for VC-HC.}
Given an instance $\Pi = (V, E, \mathbf{c}, \mathbf{m}, \mathbf{d})$ of VC-HC and an additional lower-bound $\bm{\ell}$, where $\mathbf{0} \le \bm{\ell} \le \mathbf{m}$, on the multiplicity of the vertices, we consider the following LP relaxation, with $\Psi=(V,E,\bm{\ell},\mathbf{c},\mathbf{m},\mathbf{d})$ being a parameter tuple:
%

%\smallskip

%
\begin{figure*}[h]
\centering
\begin{minipage}{0.3\textwidth}
\begin{equation}
\text{LP($\Psi$)} \notag
\end{equation}
\fbox{\begin{minipage}{.95\textwidth}
%\vspace{-2pt}
\begin{align}
\min_{(\mathbf{x},\mathbf{h})} & \quad \sum_{v \in V}x_v  & &
 \notag \\
\text{s.t.} \notag \\[-10pt]
& \quad (\mathbf{x},\mathbf{h}) \in \mathbf{Q}(\Psi) \notag
\end{align}
\vspace{-12pt}
\end{minipage}}
\end{minipage}
%\end{figure*}
\hspace{0pt}
%
%\fbox
{\begin{minipage}{0.64\textwidth}
%\small
\begin{equation}
\text{The Polytope} \enskip \mathbf{Q}(\Psi) \colon 
\notag %\tag{**}\label{polytope-VC-HC}
\end{equation}
\begin{tabular}{@{}c@{}}
\fbox{
\begin{minipage}{.9\textwidth}
\begin{subequations}
\begin{align}
& \sum_{v \in e}h_{e,v} \hspace{2pt} = \hspace{2pt} 1, & & \forall e \in E \label{LP_cdh_e} \\[2pt]
& \hspace{-4pt}\sum_{e \in E[v]}d_e\cdot h_{e,v} \hspace{2pt} \le \hspace{2pt} c_v \cdot x_v, & & \forall v \in V \label{LP_cdh_v} \\[3pt]
& \ell_v \hspace{2pt} \le \hspace{2pt} x_v \hspace{2pt} \le \hspace{2pt} m_v, & & \forall v \in V \label{LP_cdh_mv} \\[3pt]
& 0 \hspace{2pt} \le \hspace{2pt} h_{e,v} \hspace{2pt} \le \hspace{2pt} x_v,  & & \forall e \in E, \enskip v \in e  \label{LP_cdh_ev} 
\end{align}
\end{subequations}
\vspace{-14pt}
\end{minipage}\quad}
\end{tabular}
\end{minipage}}
\end{figure*}
%

%
%The two sets of variables, $h_{e,v}$ and $x_v$, correspond to the assignment of the edges and the multiplicity of the vertices, respectively.
%
%Note that, s
Since each of the variables $h_{e,v}$ and $x_v$ is bounded from both below and above, we know that $\mathbf{Q}(\Psi)$ is indeed a polytope, and the reference to its extreme points\footnote{Recall that $p \in \mathbf{Q}(\Psi)$ is an extreme point if it is not in the interior of any line segment contained in $\mathbf{Q}(\Psi)$, i.e., $p = \lambda r + (1-\lambda) s$ for some $0 < \lambda < 1$ implies that either $r \notin \mathbf{Q}(\Psi)$ or $s \notin \mathbf{Q}(\Psi)$.} is well-defined.

\medskip

Throughout this paper, for any given instance $\Pi$ of VC-HC, a number of different parameter tuples will be considered.
%$\Psi$ and the corresponding LPs are considered.
%
However, since $V$, $\mathbf{m}$, and $\mathbf{d}$ will remain the same in every considered tuple, we will simply use $(E, \bm{\ell}, \mathbf{c})$ to denote the parameter tuple $\Psi$ for the considered LP.
%by $\Psi = (E, \bm{\ell}, \mathbf{c})$.

%\mong{distinguish between the input instance and the parameter tuple}
%
%For a tuple $\Psi = (E,\ell,c)$ of parameters, the LP we consider is given as follows.
%

%

%

%%

%

%%
%%

\section{The VC Constraint and Edge Folding}
\label{sec-VC-constraint}

Our approximation algorithm is inspired by a simple one-edge example and a review on the natural LP relaxation, in particular, the constraint~(\ref{LP_cdh_ev}) which connects VC-HC to the classical vertex cover (VC) and reveals the uncapacitated problem core of VC-HC.
% inside its capacitated shell.
%
In the following we elaborate this idea in further detail.

\smallskip

Consider LP($\Psi$) with $\Psi = (E, \bm{0}, \bm{c})$, i.e., the original natural LP relaxation.
%\bm{\ell} = 0$.
% for VC-HC given above in Figure~\ref{figure-LP-VC-HC-2}.
%
It is well-understood that, although the constraints~(\ref{LP_cdh_e}) to~(\ref{LP_cdh_mv}) together with $h_{e,v} \ge 0$ for all $e,v$ already formulate VC-HC, it does not yield a solution with bounded integrality gap since it allows vertices to take arbitrarily small multiplicities.
This is illustrated by the one-edge example shown in Figure~\ref{fig-one-edge-ex}.
The constraints~(\ref{LP_cdh_e}) to~(\ref{LP_cdh_mv}) can be satisfied by setting $h_{\tilde{e},\tilde{v}} = 1$ and $x_{\tilde{v}} = \epsilon \approx 0$.
As a result, the fractional solutions can deviate arbitrarily from the integral solution for which we would have $x_{\tilde{v}} = 1$.

\smallskip

\begin{wrapfigure}[11]{r}{0.16\textwidth}
\vspace{-24pt}
\begin{flushright}
\fbox
{\begin{minipage}{.14\textwidth}
\caption{} \label{fig-one-edge-ex}
\medskip
\includegraphics[scale=1.2]{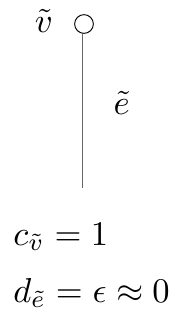}
\vspace{6pt}
\end{minipage}}
\end{flushright}
\end{wrapfigure}
This is how constraint~(\ref{LP_cdh_ev}) is introduced to bound integrality gap.
With this constraint in place, the value of $x_{\tilde{v}}$ is lower-bounded by $h_{\tilde{e},\tilde{v}}$.
This ensures that the fractional solution in this example will agree with its integral solution.
One explanation to this phenomenon is that constraint~(\ref{LP_cdh_ev}) connects this relaxation to the classical VC problem --- a problem for which we have extensively stuidied and for which we understand very well.

\medskip

\noindent
\textbf{Our key observation} here is that, while the VC constraint,~(\ref{LP_cdh_ev}), strengthens the LP relaxation, it also hints on how the rounding can be done locally when it is tight.
Suppose that $h_{e',v} = x_v$ holds for some $e'$ and $v$.
%$v \in V$ and $e \in E$. 
%
Then from~(\ref{LP_cdh_v}) we know that $$d_{e'} \cdot h_{e',v} \hspace{2pt} \le \hspace{2pt} \sum_{e\colon v \in e}d_{e}\cdot h_{e,v} \hspace{2pt} \le \hspace{2pt} c_v\cdot x_v,$$ and it follows that $d_{e'} \le c_v$ since $h_{e',v} = x_v$.

\smallskip

This suggests that:
(1) The existence of $e'$ is vital in support for the value of $x_v$, despite the fact that $d_{e'}$ can be substantially smaller than $c_v$.
(2) If the vertex $\vphantom{v^T}v$ is to be rounded up eventually, then we know that its capacity will be sufficient for covering $d_{e'}$.
This motivates the key notion of \emph{structural support} and 
%our iterative partial rounding approach.
%
the \emph{edge-folding} operation.

\begin{definition}[Supporting edge]
Let $\Psi = (\hat{E}, \bm{\ell}, \bm{c})$ be a parameter tuple and $\hat{p} = (\hat{x},\hat{h})$ be a feasible solution for LP($\Psi$).
For any $e \in \hat{E}$ and $v \in e$, we say that edge $e$ (structurally) \emph{supports} vertex $v$ in $\hat{p}$ if $0 < \hat{h}_{e,v} = \hat{x}_v$.
\end{definition}

\begin{definition}[Edge folding]
Let $\Psi = (\hat{E}, \bm{\ell}, \bm{c})$ be the current tuple, $\hat{p}=(\hat{x},\hat{h})$ be a feasible solution for LP($\Psi$), and $e$ be an edge that supports a vertex $v$ in $\hat{p}$.
By \emph{folding $e$ into $v$}, we: 

\vspace{4pt}
\begin{inparaenum}[(i)\hspace{0pt}]
\hspace{6pt}
	\item
		Remove $e$ from $\hat{E}$ and decrease $c_v$ by $d_e$. \quad% \enskip
%		
%	\item
%		Remove $e$. \quad% \enskip
%
\\[-10pt]

\hspace{6pt}
	\item
		Impose a new constraint $x_v \hspace{2pt} \ge \hspace{2pt} \hat{x}_v$ to LP($\Psi$) by setting $\ell_v = \hat{x}_v$.
\end{inparaenum}
\end{definition}

%

%\medskip

%
%\noindent
Note that, from the observation described above, folding supporting edges into the supported vertices results in no loss in the feasibility and approximation guarantee of the final solution if we know in advanced that the supported vertices are to be rounded-up eventually.
%or more precisely, $x_v \ge 1$, 
%then.
%
%In addition, the solution for the modified instance is feasible for the original input only when 
%\mong{}

%

%

%%
%%

\section{Iterative Partial Rounding for VC-HC}
\label{sec-iterative-partial-rounding-VC-HC}

The observation in the previous section motivates our iterative partial rounding approach, which focuses on tackling the VC constraint -- an uncapacitated problem core of VC-HC, using the edge-folding operation. 
At a high-level description, it is done as follows.
%
%\enskip \hspace{1pt} 
We start with the initial parameter tuple $\Psi = (E, \bm{0}, \bm{c})$.
%LP(\ref{LP_VC-HC-2}).

%\medskip

%

%\vspace{-6pt}
\begin{enumerate}
	\item
		Solve LP($\Psi$) for a \emph{basic optimal solution} ${p} = (\bm{x}, \bm{h})$.
		
		Let $I := \big\{ v \colon 0 < {x}_v < \frac{1}{f} \big\}$.
		
	\item
		If there exists an edge $e$ that 
		%is incident to some $u \in I$ while 
		structurally supports some $v \notin I$ in $p$,
		%, such that $0 < h^*_{e,v} = x^*_v$, 
		
%		\vspace{-2pt}
		\emph{then} \enskip  
		%{\small fold $e$ into $v$ by} 
%		{\color{black!90} \small // fold $e$ into $v$ }
		\vspace{-18pt}
		\begin{itemize}[\qquad \quad -- ]
			\item
				Fold $e$ into $v$.
			\item
				Restart Step 1.
		\end{itemize}
		
		\vspace{-40pt}
		\hspace{5.6cm}
		\begin{minipage}{.6\textwidth}
		\emph{otherwise} 
		\vspace{-12pt}
		\begin{itemize}[\qquad\qquad\quad -- ]
			\item
				Round up $\bm{x}$ and stop.
		\end{itemize}
		\end{minipage}
\vspace{6pt}
\end{enumerate}
In contrast to previous approaches, e.g.,~\cite{Saha:2012:SCR:2359332.2359443,doi:10.1137/1.9781611973402.124,1151271}, which round up every vertex with decent value at once, we only deal with those that are structurally supported.
Furthermore, instead of rounding up the variables aggressively to one, we round them partially to a fractional value by imposing stronger lower bound constraints on them.
This nature allows previously non-supporting edges to become supporting as the algorithm iterates.
The key feature of this approach is that it offers a series of LPs with {non-increasing} objective values which preserves the optimality of the initial LP for the final rounding step.

\medskip

Let $\hat{\Psi} = (\hat{\MC{E}}, \bm{\hat{\ell}}, \hat{\bm{c}})$ denote the parameter tuple when the algorithm enters the final rounding step.
%
%Let $\hat{\MC{E}}$ denote the set of edges that are folded during the process and 
%
Let $\bm{h}'$ denote the $\{0,1\}$-assignment function for $E \setminus \hat{\MC{E}}$ that indicates the vertex which each edge $e \in E \setminus \hat{\MC{E}}$ is folded into.
%, i.e.,
% into.
%
In particular, for each $e\in E \setminus \hat{\MC{E}}$ and $v \in e$, $h'_{e,v}$ is set to $1$ if $e$ is folded into $v$ and $0$ otherwise.
%for the particular vertex which $e$ is folded into and zero for the remaining.
%
Let $\hat{p} = (\hat{x}, \hat{h})$ denote the basic optimal solution computed for $\hat{\Psi}$.

\smallskip

The final output $(\bm{x}^*, \bm{h}^*)$ is defined as follows.
For any $v \in V$ and $e \in E[v]$, let
$$x^*_v := \CEIL{\vphantom{\bigcup} \hat{x}_v}
\quad \text{and} \quad
h^*_{e,v} := \begin{cases}
\hspace{2pt} \hat{h}_{e,v}, \enskip & \text{if $e \in \hat{\MC{E}}$,} \\
\hspace{2pt} h'_{e,v}, & \text{otherwise.}
\end{cases}
$$
Then $(\bm{x}^*, \bm{h}^*)$ is output as the solution.

\medskip

A pseudo-code of this algorithm is provided in Figure~\ref{fig-overview-tight-vc-hc} in Page~\pageref{fig-overview-tight-vc-hc} for further reference.
We remark that, since this approach is essentially insensitive to multiple foldings, in the actual algorithm we will fold every supporting edge, provided that it supports some $v \notin I$.
Furthermore, ties are broken arbitrarily if an edge supports multiple vertices outside $I$.
%
%In Section~\ref{sec-itr-for-VC-HC} we provide a
%A formal description of the algorithm Tight-VC-HC is also provided in Section~\ref{sec-itr-for-VC-HC} in the appendix for further reference.
%

\medskip

Let Tight-VC-HC denote the above algorithm. 
Our tight approximation result is stated in the following theorem:

\begin{theorem}%[\bf Restate of Theorem~\ref{thm-informal-tight-vc-hc}]
\label{thm-tight-vc-hc}
On any instance $\Pi = (V,E,\bm{c},\bm{m}, \bm{d})$ of VC-HC with maximum edge size $f \ge 2$, algorithm Tight-VC-HC outputs an $f$-approximation $(\bm{x}^*, \bm{h}^*)$ in polynomial time.
\end{theorem}

From the usage of edge-folding operation in the algorithm and the intuition provided in Section~\ref{sec-VC-constraint} for the VC constraints, it is not difficult to see that:
\begin{itemize}
	\item
		Algorithm Tight-VC-HC terminates in $O(|E|)$ rounds.
		
	\item
		The output $(\bm{x}^*, \bm{h}^*)$ is indeed feasible for the initial LP$\big( (E, \bm{0}, \bm{c}) \big)$, and $\bm{x}^*$ is integral.
\end{itemize}

Due to the page limit, the proof for the feasibility of algorithm Tight-VC-HC is provided in Section~\ref{sec-tight-vc-hc-feasibility} for further reference.
In order to fully-verify the detail provided in Section~\ref{sec-tight-vc-hc-feasibility}, we also refer the reader to the notations and properties given in the beginning of Section~\ref{sec-proof-thm-tight-vc-hc}.
%
%
%\smallskip
%
%
In the remaining of this paper, we will describe how the approximation guarantee 
%can be
is established.
%
%For a complete reference to the approximation guarantee, see also Section~\ref{sec-tight-vc-hc-analysis}.
%
%\mong{however .......*****}
%\mong{for a complete reference of this proof ****}

\newpage

\section{Approximation Guarantee}
\label{sec-approx-ratio}

In this section we establish our approximation guarantee.
The argument we use builds on the fact that $p = (\bm{x}, \bm{h})$ is a basic feasible solution for the extended LP relaxation LP($\Psi$).
%
%
%\smallskip
%
%
In Section~\ref{sec-separation-lemma-notation} we first define the notions that will be used to capture the structural properties given by our partial rounding approach.
In Section~\ref{sec-separation-lemma} we formally state our separation lemma, which is our main technical tool for establishing the approximation guarantee, and discuss the intuitions behind.
%
%We also note that, although the technical lemma is presented in an abstract way, making minimal assumption it needs, it is more intuitive to relate the concept to our rounding algorithm.
%
%
%\smallskip
%
%
In Section~\ref{subsec-approx-ratio}, we will describe how the technical separation lemma 
%can be
is used to establish the tight approximation guarantee.

\subsection{Notion and Definitions for Extremality}
\label{sec-separation-lemma-notation}

In this section we introduce notions that are required to state our separation lemma for the polytope $\mathbf{Q}(\Psi)$.
Let $p = (\mathbf{x},\mathbf{h})$ be a point in $\mathbf{Q}(\Psi)$, where $\Psi = (\hat{E}, \bm{\ell}, \bm{c})$ is a parameter tuple.
The following terminologies are defined with reference to point $p$ and parameter tuple $\Psi$.

\medskip

\begin{definition}[Non-extremal vertex]
For any $v \in V$, we say that $v$ is \emph{non-extremal} if $\ell_v < x_v < m_v$. Otherwise $v$ is said to be \emph{extremal}.
\end{definition}

Intuitively, a vertex is non-extremal if and only if $x_v$ is not constrained by~(\ref{LP_cdh_mv}) in $\mathbf{Q}(\Psi)$.
For any $e \in \hat{E}$ and $v \in e$, we say that $v$ is an \emph{active end} of $e$ if $h_{e,v} > 0$.

\begin{definition}[Active subedge and active edge sets]
For any $e \in \hat{E}$, we define the \emph{active subedge} of $e$, denoted $e^{\OP{actv}}_h$, to be the set of its active ends, i.e., 
$$e^{\OP{actv}}_h := \Big\{ \vphantom{\dfrac{}{}}\hspace{2pt} v \hspace{2pt} \colon \hspace{2pt} v\in e, \enskip h_{e,v} > 0 \hspace{2pt} \Big\}.$$
% to denote the active ends of $e$.
%
%Note that we have $e^{\OP{actv}}_h \subseteq e$ for all $e$.
%
For any edge subset $\MC{E} \subseteq \hat{E}$, we extend the above definition and use 
$\MC{E}^{\OP{actv}}_h := \BIGBP{ \vphantom{\dfrac{}{}}\hspace{2pt} e^{\OP{actv}}_h \hspace{2pt} \colon \hspace{2pt} e \in \MC{E} \hspace{2pt}}$ to denote the set of active subedges of the edges in $\MC{E}$.
\end{definition}

Intuitively, the active subedge corresponds to the set of vertices to which the demand of an edge is assigned in $\bm{h}$.
Note that, since $p \in \mathbf{Q}(\Psi)$, it follows from constraint~(\ref{LP_cdh_e}) that $e^{\OP{actv}}_h \neq \emptyset$ for all $e \in \hat{E}$.
Recall that we also use $\MC{E}[U]$ for a vertex subset $U$
% \subseteq V$ 
of $V$ to denote the set of edges in $\MC{E}$ that are incident to some vertex in $U$.
Given the definition of active edge sets, 
%we remark that 
$\vphantom{{\dfrac{}{}}^{\frac{T}{T}}} \MC{E}^{\OP{actv}}_h[U]$ is then used to denote the set of incident (sub)edges $U$ has in $\MC{E}^{\OP{actv}}_h$, i.e., 
$$\MC{E}^{\OP{actv}}_h[U] := \Big\{\hspace{2pt} \tilde{e} \hspace{2pt} \colon \hspace{2pt} \tilde{e} \in \MC{E}^{\OP{actv}}_h \text{ such that } \tilde{e} \cap U \neq \emptyset \hspace{2pt} \Big\}.$$
Intuitively, this corresponds to the set of active subedges from $\MC{E}$ that intersect $U$.
Note that it is a collection of subedges rather than the original edges.
%

%\medskip

%
%Consider an edge $e \in E$ and a vertex $v \in e$. 
%
%Recall that we say $e$ structurally {supports} $v$ if $0 < h_{e,v} = x_v$.
%
%In this case $e$ is also referred to as a supporting edge (for $v$).
%

%\smallskip

\begin{definition}[Supporting sets and supported sets]
For a vertex subset $U \subseteq V$, we say that
\begin{itemize}
	\item
		$U$ is \emph{supporting} if there exists an edge $e \in \hat{E}$ with $e^{\OP{actv}}_h \cap U \neq \emptyset$ and a vertex $v \in e^{\OP{actv}}_h \setminus U$ such that $e$ supports $v$.
		
	\item
		$U$ is \emph{supported} if there exists $e \in \hat{E}$ and $v \in e^{\OP{actv}}_h \cap U$ such that 
		%$\tilde{e} \in E^{\OP{actv}}_h[v]$ such that 
		$v$ is supported by $e$.
\end{itemize}
\end{definition}

The definition of supporting sets and supported sets may seem unnatural in the beginning. 
However, consider the set $I := \big\{ v \colon 0 < {x}_v < \frac{1}{f} \big\}$ of vertices with small fractional values.
In our partial rounding approach, we iteratively remove the support from vertices with decent fractional values, i.e., those not in $I$, and when the algorithm terminates, we know that the set of small vertices is not supporting while the set of vertices with large fractional values, e.g., those with $x_v > 1$, is not supported.

\subsection{A Separation Lemma for Polytope $\mathbf{Q}(\Psi)$}
\label{sec-separation-lemma}

We formally state our structural separation lemma, which states that a strong partition exists when the edge-to-vertex supporting relations are properly eliminated:

\begin{lemma}[Existence of a strong partition] \label{thm-structural-mapping}
Let $p = (\mathbf{x},\mathbf{h})$ be an extreme point of polytope $\mathbf{Q}(\Psi)$.
For any disjoint sets $\MC{I}$, $\MC{D}$ of non-extremal vertices,
if $\MC{I}$ is not supporting and $\MC{D}$ is not supported, 
then there exists a mapping\hspace{2pt}\footnote{We would like to refer the reader to the previous section for the definitions of non-extremal vertices, active subedges, and supporting(-ed) sets in order to fully-access the structural message in this lemma.}
\vspace{-4pt}\enskip $$\Gamma \hspace{2pt} \colon \enskip \MC{I} \enskip \mapsto \enskip E^{\OP{actv}}_h[\MC{I}] \setminus E^{\OP{actv}}_h[\MC{D}]
\vspace{-6pt}$$ %\enskip
such that the following hold for any $v \in \MC{I}$: 
\vspace{-2pt}
\begin{enumerate}
	\item
		(reflexive) \enskip $v \in \Gamma(v)$. %$\Gamma(v) \in E^{\OP{actv}}_h[v] \setminus E^{\OP{actv}}_h[\MC{D}]$.
		
		\vspace{-2pt}
	\item
		(closed under intersection) \enskip $\Gamma(u)\cap \Gamma(v) \hspace{2pt} \subseteq \hspace{2pt} \MC{I}$ \enskip  for any $u \in \MC{I} \setminus \{v\}$.
\end{enumerate}
\end{lemma}

\medskip

Given sets $\MC{I}$ and $\MC{D}$ as stated in the prerequisite, Lemma~\ref{thm-structural-mapping} says that, there exists a mapping $\Gamma$ such that, for each $v \in \MC{I}$, the active subedge $\Gamma(v)$ contains the vertex $v$ but 
%$E^{\OP{actv}}_h[v]$ that contains 
no vertices in $\MC{D}$.
Furthermore, the intersection of these active subedges can happen only inside $\MC{I}$, i.e., 
%That is, 
they are mutually disjoint outside the set $\MC{I}$.
% provided that the intersection is non-empty.
%
See also Figure~\ref{fig-illustration-separation-lemma-general} for an illustration.

\smallskip

Note that, the statement of Lemma~\ref{thm-structural-mapping} does not exclude the possibility that $\Gamma(v) \subseteq \MC{I}$ and the mapping $\Gamma$ may not necessarily be injective.
In fact, it depends on the choice of $\MC{I}$, since we make no prior assumption on the edge-vertex supporting relation inside $\MC{I}$.
When all non-extremal vertices are selected into $\MC{I}$, then Figure~\ref{fig-illustration-separation-lemma-general} is what we can expect.

\smallskip

\begin{figure*}[tp]
\centering
%\vspace{10pt}
\includegraphics[scale=0.6]{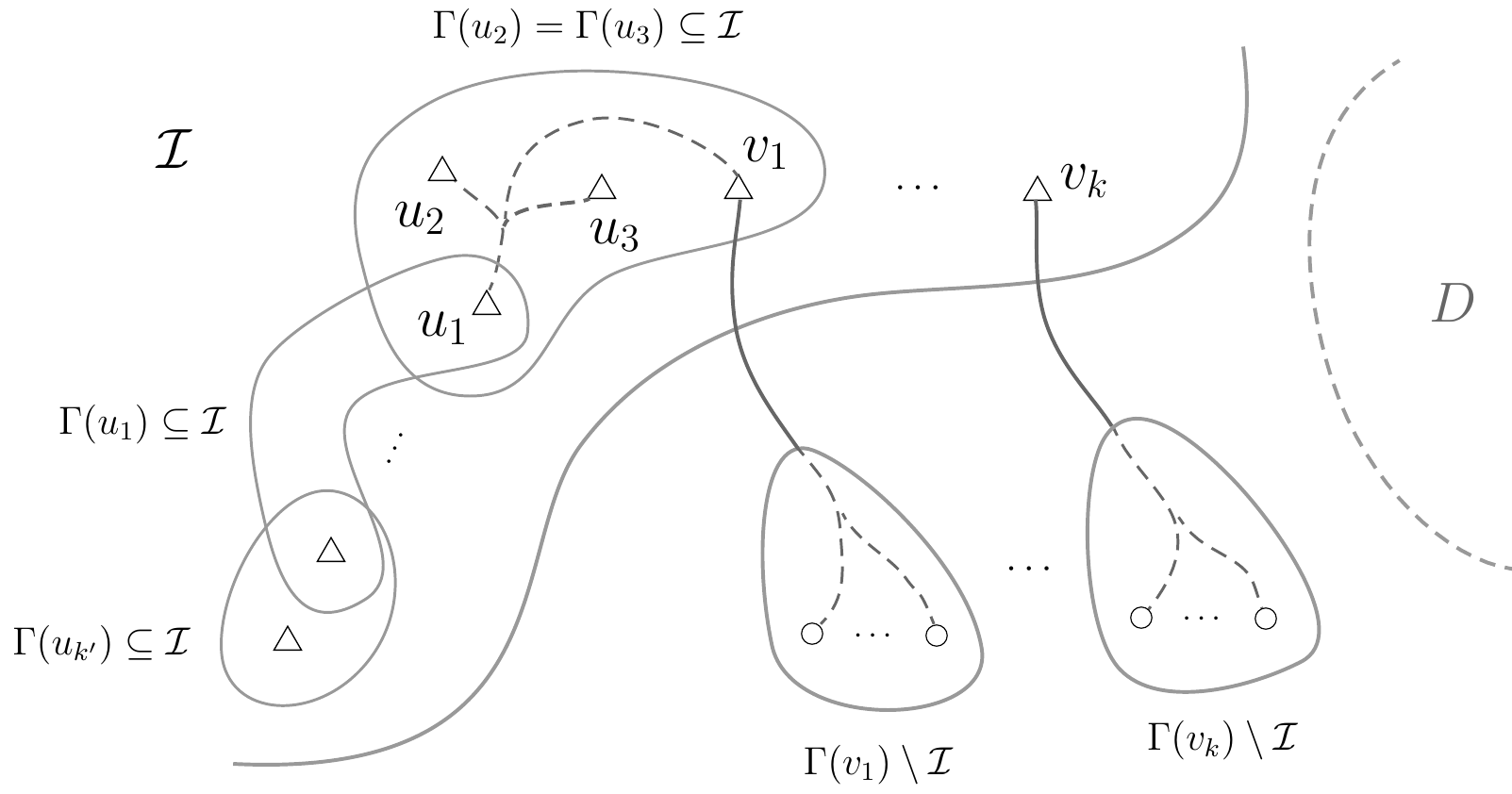}
\vspace{-6pt}
\caption{$_{\dfrac{}{}}$ \hspace{-8pt} A depiction on the set $\MC{I}$ and the subedges given by $\Gamma$. $\big\{u_i\big\}_{1\le i\le k'} \subseteq \MC{I}$ denote vertices with $\Gamma(u_i) \subseteq \MC{I}$ and $\big\{v_i\big\}_{1\le i\le k}$ denote vertices with $\Gamma(v_i) \not\subseteq \MC{I}$, i.e., $\Gamma(v_i)\setminus \MC{I} \neq \emptyset$.}
\label{fig-illustration-separation-lemma-general}
\end{figure*}

With an appropriate choice of $\MC{I}$, e.g., one that ensures $\tilde{e} \not\subseteq \MC{I}$ for all $\tilde{e} \in E^{\OP{actv}}_h[\MC{I}]$, however, the injective property of $\Gamma$ will follow from the two properties stated in Lemma~\ref{thm-structural-mapping}.
% and thereby defines a partition for vertices in $V \setminus \MC{I}$.
%
Notice that the set of vertices with small fractional values is exactly one of such choices.
In Section~\ref{subsec-approx-ratio} we will illustrate how the mapping $\Gamma$ is used to obtain a tight approximation result.

\smallskip

In the following we discuss the technique we use to build this lemma and the intuitions behind.
A formal proof to this lemma is provided in Section~\ref{sec-proof-separation-lemma} for further reference.

\paragraph{Extreme points of the VC-HC polytope.}
Lemma~\ref{thm-structural-mapping} can be seen as a recast of the well-established Carath\'{e}odory's theorem~\cite{Schrijver:1986:TLI:17634} in the language of VC-HC polytope.
Intuitively, it says that: A vertex-to-edge correspondence with strong separation guarantee exists at places where the the VC constraint is inactive.

%\smallskip

%
%We have seen similar ideas been shared~\cite{doi:10.1137/1.9781611973402.124} in the usage of the multi-set multi-cover polytope, which is essentially a vertex-to-vertex matching structure, to extract a vertex-to-vertex matching.
%
%When the VC-HC polytope is considered, however, we have a substantially different entangling relation to deal with, in particular, relation between the vertices $V$, the edges $E$, and their cross-products $E\times V$. 
%
%Lemma~\ref{thm-structural-mapping} provides an interpretation from the perspective of Carath\'{e}odory's theorem to this situation:
%
%By properly removing the entangled $E \times V$ relation outside $\MC{I}$,
% while assuming a reasonable entanglement outside $I$, 
%we can extract a strong and meaningful property for our rounding approach.
%

\smallskip

Below we further elaborate this idea.
% in detail.
%
Let $p = (\bm{x},\bm{h})$ be an extreme point of the VC-HC polytope and $\OP{Var}(\MC{I}, \MC{D})$ denote the set of 
%non-extremal 
variables that are related to $\MC{I}$, $\MC{D}$ and 
%its incident 
the edges in $E[\MC{I}]$.
By Carath\'{e}odory's theorem, 
%each set of non-extremal variables in $p$, including $\OP{Var}(\MC{I})$,
the variables in $\OP{Var}(\MC{I}, \MC{D})$ are constrained by a set, with the same cardinality, of linearly independent inequalities
%from the the LP(\ref{LP_VC-HC-2}) 
that hold with equality.

\smallskip

This gives a one-to-one (injective) correspondence between the variables and the constraints.
From the assumption that $\MC{I}$ is non-extremal and not supporting, it follows that:
\begin{enumerate}[(i) ]
	\item
		The variable $x_v$ \emph{must} be constrained by either (\ref{LP_cdh_v}) or (\ref{LP_cdh_ev}), for all $v \in \MC{I}$.
	\item
		The variable $h_{e,u}$, where $h_{e,u} \in \OP{Var}(\MC{I},\MC{D})$ such that $e \cap \MC{I} \neq \emptyset$ and $u \notin \MC{I}$, \emph{can only} be constrained by (\ref{LP_cdh_e}) and (\ref{LP_cdh_v}).
\end{enumerate}

From (i) and the fact that the constraints extracted by Carath\'{e}odory's theorem are linearly independent, we show that for each $v \in \MC{I}$, an unique edge constraint, i.e., (\ref{LP_cdh_e}), can be identified.
This gives the edge for 
%definition of the active subedge 
$\Gamma(v)$ to be defined.
Then the main claim of this lemma is guaranteed by (ii) and the injective property of the variable-to-constraint mapping.

\smallskip

To see that the vertices in 
%$\vphantom{{\dfrac{}{}}^{\frac{T}{T}}} D := \big\{\hspace{2pt} v \hspace{1pt} \colon \hspace{1pt} 1 < x^*_v < m(v) \hspace{2pt} \big\}$ is 
$\MC{D}$ can be excluded from the subedges defined above, it suffices to observe that for any $\vphantom{\frac{}{}} v \in \MC{D}$, the only constraint for variable $x_v$ to be constrained is (\ref{LP_cdh_v}), i.e., $v$ itself, since it is non-extremal and not supported by definition.
Therefore it cannot be included in any of these active subedges since their corresponding constraints have already been occupied.
The complete proof of this lemma is provided in Section~\ref{sec-proof-separation-lemma} for further reference.

\subsection{Establishing the Approximation Guarantee}
\label{subsec-approx-ratio}

In this section we describe how Lemma~\ref{thm-structural-mapping} is used to obtain our tight approximation guarantee.
Let $\hat{p} = (\hat{x}, \hat{h})$ denote the basic optimal solution computed in the final iteration of the algorithm.
Consider the two sets
$$I := \Big\{ \hspace{2pt} v \hspace{2pt} \colon \hspace{2pt} v \in V, \enskip 0 < \hat{x}_v < \frac{1}{f} \hspace{2pt} \Big\} \quad \text{and} \quad D := \Big\{ \hspace{2pt} v \hspace{2pt} \colon \hspace{2pt} v \in V, \enskip 1 < \hat{x}_v < m_v \hspace{2pt} \Big\}.$$
%
%
%\smallskip
%
%
%Note that, b
%By the definition of $I$, i
It follows that $e^{\OP{actv}}_h \not\subseteq I$ for all $e$ such that $e^{\OP{actv}}_h \cap I \neq \emptyset$.
Therefore condition~(ii) of Lemma~\ref{thm-structural-mapping} implies that $\Gamma$ will be injective, i.e., $\Gamma(u) \neq \Gamma(v)$ whenever $u \neq v$,
and thereby defines an equivalence relation in $V \setminus I$ with respect to $I$, witnessed by the active subedge $\Gamma(v)$ for each $v \in I$.
% partition for vertices in $V \setminus \MC{I}$.
%
See also Figure~\ref{fig-illustration-partition-tech-intro} for an illustration of this partition.

\begin{figure*}[bp]
%\begin{wrapfigure}[14]{r}{0.45\textwidth}
%\vspace{-8pt}
%\begin{flushright}
%\fbox
%{\begin{minipage}{.43\textwidth}
\centering
\qquad\qquad\includegraphics[scale=0.6]{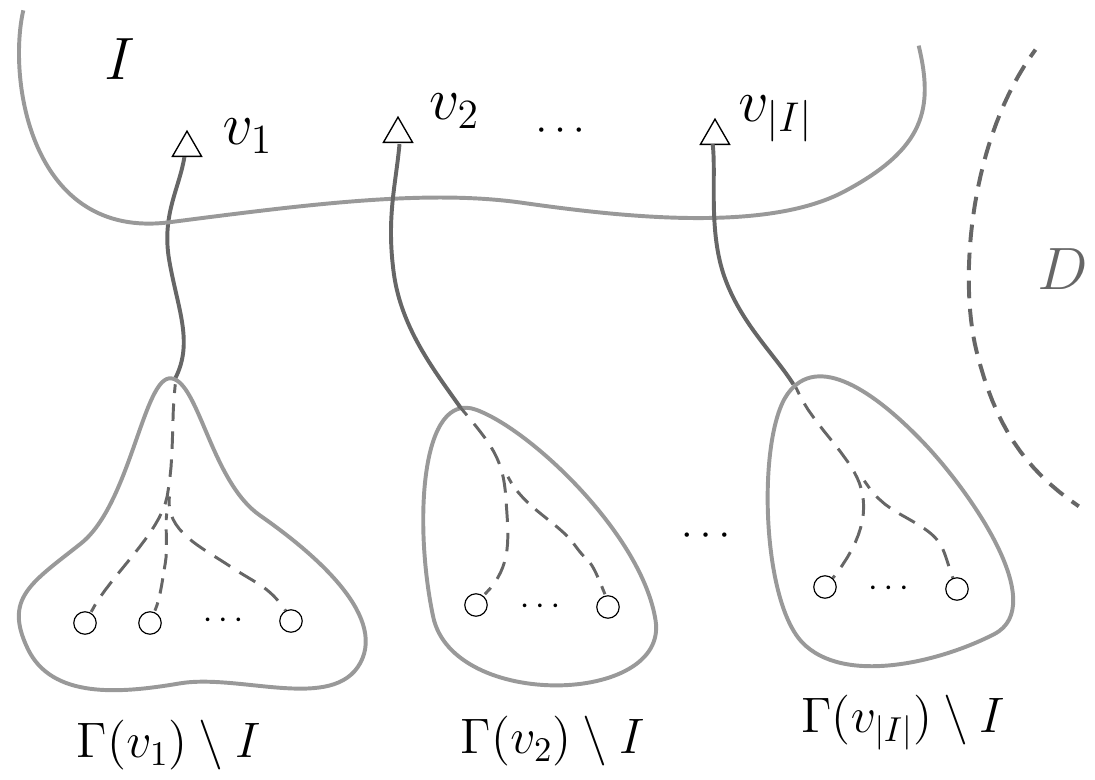}
\vspace{-6pt}
\caption{A partition of vertices in $V \setminus I$, witnessed by the active subedges $\big\{ \Gamma(v_i) \big\}_{1\le i\le |I|}$}
\label{fig-illustration-partition-tech-intro}
%\vspace{-16pt}
%\end{minipage}}
%\end{flushright}
%\end{wrapfigure}
\end{figure*}

\smallskip

Consider the active subedge $\Gamma(v)$ for some $v \in I$.
By Lemma~\ref{thm-structural-mapping}, we know that $\Gamma(v)$ contains no vertices in $D$, since $\Gamma(v) \notin E^{\OP{actv}}_h[D]$.
It follows that, for each $u \in \Gamma(v)$, either $\hat{x}_u = m_u$ or $\hat{x}_u \le 1$ will hold.
For the former case, the rounding cost of $v$ can be absorbed by $u$, since we have $m_u \ge 1$ and $f \ge 2$.
Below we consider the case that $\hat{x}_u \le 1$ for all $u \in \Gamma(v)$.

\smallskip

Suppose that $\hat{x}_u \le 1$ for all $u \in \Gamma(v)$.
Let $e$ be an edge whose active subedge is exactly $\Gamma(v)$.
From the constraint~(\ref{LP_cdh_e}) we have $\sum_{u \in \Gamma(v)} \hat{h}_{e,u} = 1$.
%, where $e$ is one of the particular edges whose active subedge is exactly $\Gamma(v)$.
%
Since we know that $\hat{x}_u$ is small for all $u \in \Gamma(v) \cap I$,  it follows that 
\vspace{-2pt}
$$\hat{x}_v + \sum_{u \in \Gamma(v) \setminus I}\hat{x}_u \enskip \ge \enskip 1 - \hspace{-10pt} \sum_{u \in (\Gamma(v) \cap I) \setminus \{v\}} \hat{x}_u \enskip \ge \enskip 1 - \frac{1}{f}\cdot \Big| \hspace{2pt} (\Gamma(v) \cap I) \setminus \{v\} \hspace{2pt} \Big| \enskip \ge \enskip \frac{1}{f}\cdot \Big| \hspace{2pt}  \big( \Gamma(v) \setminus I \big) \cup \{v\} \hspace{2pt} \Big|,
\vspace{-4pt}$$
where in the last inequality we use the fact that $\big|(\Gamma(v) \cap I) \setminus \{v\}\big| = \big|\Gamma(v)\big| - \big|(\Gamma(v) \setminus I) \cup \{v\}\big|$ and $\big|\Gamma(v)\big| \le \big| e \big| \le f$.
%, it follows that 
%$$1 - \frac{1}{f}\cdot \Big| \hspace{2pt} (\Gamma(v) \cap I) \setminus \{v\} \hspace{2pt} \Big| \enskip ,$$ 
%and 
Therefore the rounding cost incurred by $v$ and vertices in $\Gamma(v) \setminus I$ can be bounded within $\vphantom{{\frac{}{}}^{\frac{T}{T}}} f\cdot \big( \hat{x}_v + \sum_{u \in \Gamma(v) \setminus I} \hat{x}_u \big)$, for any $f \ge 2$.

\medskip

We remark that the exclusion of $D$ from the image of $\Gamma$ is crucial in obtaining a tight approximation for $f=2$, i.e., the normal graphs.
%, however, it requires further structural property of the partition.
%
The reason is that, rounding up medium-sized vertices, e.g., $u$ with $1<\hat{x}_u < 1+1/2$, already results in a factor of $2$, rendering them unable to absorb additional rounding cost incurred by small vertices.
%
%
%To be more precise, vertices with medium-sized multiplicities, e.g., $u$ with $1<\hat{x}_u < 1+1/2$, must be excluded from the image of $\Gamma(v)$, since rounding up these vertices already result in a factor of $2$ and they will not be able to absorb the rounding cost for vertices in $I$.
%
%This is how the set $\vphantom{{\dfrac{}{}}^{\frac{T}{T}}} D := \big\{\hspace{2pt} v \hspace{1pt} \colon \hspace{1pt} 1 < \hat{x}_v < m(v) \hspace{2pt} \big\}$ is defined.
%
%Since $D$ is not supported and also disjoint from $I$, by Lemma~\ref{thm-structural-mapping} we know that $\Gamma(v) \cap D = \emptyset$ for all $v \in I$.
%
%Therefore a tight $2$-approximation is obtained for normal graphs as well.
%
%
%\smallskip
%
%
A formal proof for the approximation guarantee is provided in Section~\ref{sec-tight-vc-hc-analysis} for further reference.

\smallskip

\section*{Acknowledgement}

The author would like to thank Kai-Min Chung, Herbert Yu, and the anonymous referees for their very helpful comments on the presentation of this work.

\vfill

\begin{figure*}[htp]
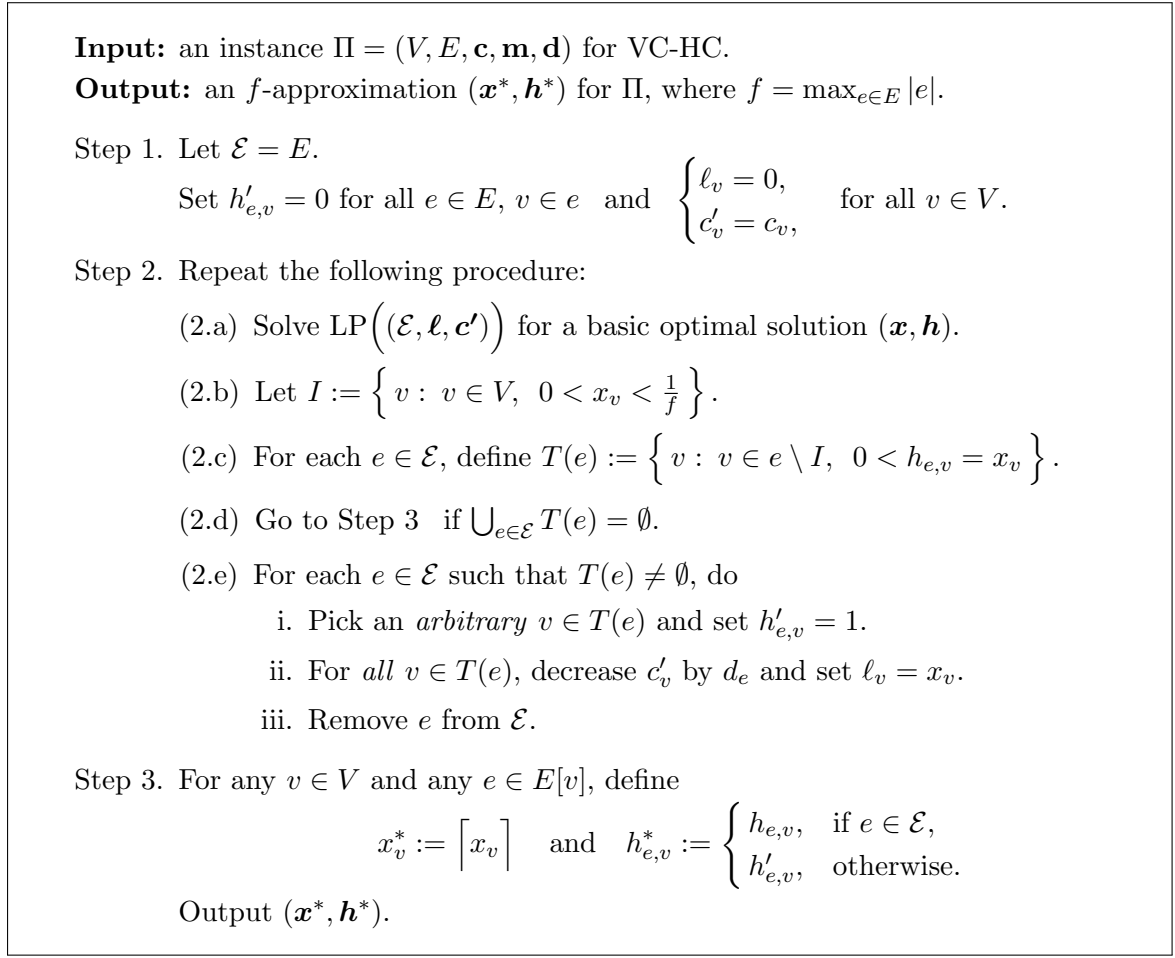

%\noindent
\fbox
{\qquad\begin{minipage}{0.9\textwidth}
\smallskip
\medskip
%Algorithm Tight-VC-HC:
%\mong{Input, Output} 

{\bf Input:} an instance $\Pi = (V, E, \mathbf{c}, \mathbf{m}, \mathbf{d})$ for VC-HC.

\vspace{2pt}

{\bf Output:} an $f$-approximation $(\bm{x}^*, \bm{h}^*)$ for $\Pi$, where $f = \max_{e \in E}|e|$.
\begin{enumerate}[Step 1.]
	\item
		Let $\MC{E} = E$.
		
		\vspace{-8pt}
		Set $h'_{e,v} = 0$ for all $e \in E$, $v \in e$ \enskip and \enskip
		$\begin{cases}
			\ell_v = 0, \\
			c'_v = c_v,
		\end{cases}$ for all $v \in V$.

		\vspace{-4pt}
		
	\item \label{alg-tight-vc-hc-loop-step}
		Repeat the following procedure:
		\vspace{-4pt}
		\begin{enumerate}[(\ref{alg-tight-vc-hc-loop-step}.a)]
			\item
				Solve LP$\Big((\MC{E}, \bm{\ell}, \bm{c'})\Big)$ for a basic optimal solution $(\bm{x}, \bm{h})$.
				
			\item
				Let $I := \BIGBP{\hspace{2pt} v \hspace{2pt} \colon \hspace{2pt} v \in V, \enskip 0 < x_v < \frac{1}{f} \hspace{2pt}}.$
				
			\item
				For each $e \in \MC{E}$, define $T(e) := \BIGBP{ \vphantom{\dfrac{}{}}\hspace{2pt} v \hspace{2pt} \colon \hspace{2pt} v \in e \setminus I, \enskip 0 < h_{e,v} = x_v \hspace{2pt}}.$
				
				\smallskip
				
			\item
				Go to Step~\ref{alg-tight-vc-hc-after-loop} \enskip if 
				$\bigcup_{e \in \MC{E}}T(e) = \emptyset.$
				
				\smallskip
				
			\item
%				/* Fold in supporting edges */
%				
%				\smallskip
%				
				For each $e \in \MC{E}$ such that $T(e) \neq \emptyset$, do
				\begin{enumerate}
					\item
						Pick an \emph{arbitrary} $v \in T(e)$ and set $h'_{e,v} = 1$.
						
						\smallskip
						
					\item
						For \emph{all} $v \in T(e)$, 
						decrease $c'_v$ by $d_e$ and set $\ell_v = x_v$.
						%$\begin{cases}
						%\text{decrease $c'(v)$ by $d(e)$, and} \\
						%\text{set $\ell(v) = x_v$.}
						%\end{cases}$
						
						\smallskip
						
					\item
						Remove $e$ from $\MC{E}$.
				\end{enumerate}
		\end{enumerate}
		
%		\smallskip
		
	\item \label{alg-tight-vc-hc-after-loop}
		For any $v \in V$ and any $e \in E[v]$, define 
		\vspace{-10pt} 
		$$x^*_v := \CEIL{\vphantom{\bigcup} x_v}
		\quad \text{and} \quad
		h^*_{e,v} := \begin{cases}
			\hspace{2pt} h_{e,v}, \enskip & \text{if $e \in \MC{E}$,} \\
			\hspace{2pt} h'_{e,v}, & \text{otherwise.}
		\end{cases}
		\vspace{-8pt}
		$$
%
%		
%		\smallskip
%		
		Output $(\bm{x}^*, \bm{h}^*)$.
\end{enumerate}
\vspace{-2pt}
\end{minipage}}
\caption{Algorithm Tight-VC-HC.}
\label{fig-overview-tight-vc-hc}
\end{figure*}

%

%%
%%

%%%
%%%

\bibliographystyle{plain}
\bibliography{approx_capacitated_domination}

\newpage

\begin{appendix}

\section{Algorithm Tight-VC-HC}
\label{sec-itr-for-VC-HC}

In this section, we formally describe our tight approximation algorithm
%algorithm Tight-VC-HC 
for VC-HC.
% and prove Theorem~\ref{thm-tight-vc-hc}:

%

%\begin{theorem}%[\bf Restate of Theorem~\ref{thm-informal-tight-vc-hc}]
%\label{thm-tight-vc-hc}
%On any instance $\Pi = (V,E,c_v,m_v,d_e)$ of VC-HC with maximum edge size $f$, algorithm Tight-VC-HC outputs an $f$-approximation $(x^*,h^*)$ in polynomial time.
%\end{theorem}

%

%\medskip

%%
%%

\begin{figure*}[t]
\begin{tikzpicture}[auto,node distance = 4cm,>=latex']
%\tikzstyle{line} = [draw, -latex']
%
\tikzstyle{input} = [ellipse, text centered, fill=black!10, inner sep=4pt];
%trapezium, trapezium left angle=100, trapezium right angle=80, text centered, draw, fill=red!20];
%
\tikzstyle{block} = [rectangle, fill=black!5, text centered, rounded corners, inner sep=8pt];
\tikzstyle{decision} = [diamond, aspect=2.2, fill=black!5, text centered];
\tikzstyle{arrow} = [thick,->,>=stealth]
\node (input) [input] {initial $\Psi$};
\node (loop) [block, right of=input, xshift=-0.8cm] {\parbox{6em}{\centering Solve LP($\Psi$) \\ for \\ basic opt $p$}};
\node (decide) [decision, right of=loop, xshift=0.4cm] {$\enskip \bigcup_{e\in \MC{E}} T(e) \enskip$};
\node (done) [block, right of=decide, xshift=1cm, fill=black!10] {\parbox{6em}{\centering Round up $p$ \\ \& Finish.}};
\node (revise) [block, below of=decide, yshift=1.5cm] {Refine $\Psi$: fold $e$ into $T(e)$};
\draw [arrow] (input) -- (loop);
\draw [arrow] (loop) -- (decide);
\draw [arrow] (decide) -- node {$=\emptyset$} (done);
\draw [arrow] (decide) -- node {$\neq \emptyset$} (revise);
\draw [arrow] (revise) -| (loop);
\end{tikzpicture}
\caption{Overview of Algorithm Tight-VC-HC.}
\label{fig-flowchart-tight-vc-hc}
\end{figure*}
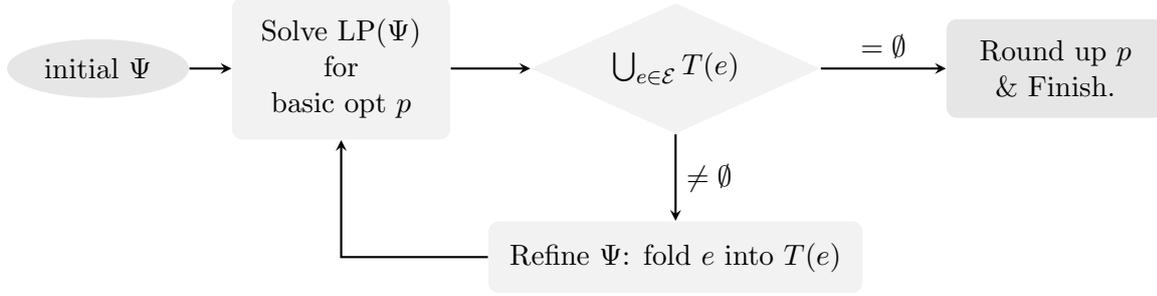

\smallskip

%\paragraph{Description of algorithm Tight-VC-HC.}
%
Let $\Pi = (V,E,c_v,m_v,d_e)$ denote the input instance of VC-HC.
In the algorithm we start with a trivial initial parameter tuple $\Psi = (E,\mathbf{0},\bm{c})$.
% of parameters for the relaxation described in LP(\ref{ILP_cdh}).
%
As the algorithm iterates, the tuple $\Psi$ is refined gradually until we have a nice parameter tuple which makes Lemma~\ref{thm-structural-mapping} applicable.
In this case we round up the solution given by LP($\Psi$) and stop.
An overview of the algorithm is given in Figure~\ref{fig-flowchart-tight-vc-hc}.
%
%For a complete reference of the algorithm see also Figure~\ref{fig-overview-tight-vc-hc}.
%
Below we describe the algorithm in details.
% and explain for the intuitions behind.
%

\smallskip

For a given parameter tuple $\Psi = (\MC{E}, \bm{\ell}, \bm{c'})$, the algorithm solves LP($\Psi$) for a basic optimal solution $p = (\bm{x},\bm{h})$.
Consider the set $I := \BIGBP{v\colon v \in V, \enskip 0 < x_v < \frac{1}{f}}$.
%
%Consider the following two sets $U := \BIGBP{v \colon v \in \MC{V}, \enskip \frac{1}{f} \le x_v \le 1}$ and $W := \BIGBP{v \colon v \in \MC{V}, \enskip 1 < x_v < 1+\frac{1}{f}}.$
%
%Since $h_{e,v} \le 1$ for all $e,v$, we know that, any supporting edge must be supporting vertices in $U \cup I$.
%
For any $e \in \MC{E}$, let 
$$T(e) := \BIGBP{\vphantom{\bigcup} \hspace{2pt} v \hspace{3pt} \colon \hspace{3pt} v \in e\setminus I, \enskip 0 < h_{e,v} = x_v \hspace{2pt}}$$
denote the set of vertices in $V \setminus I$ which $e$ is supporting.
%
%
%\smallskip
%
%
If $\bigcup_{e \in \MC{E}}T(e)$ is empty, then no refinement is required %and Theorem~\ref{thm-structural-mapping} is applicable 
and we proceed to the rounding stage.
Otherwise, $\Psi$ is refined as follows:
%is refined via the following edge-folding process:
% (a) edge-folding and (b) vertex-splitting.
%

\begin{itemize}
	\item
		\emph{Fold in supporting edges.}
		\enskip
		For each $e \in \MC{E}$ such that $T(e) \neq \emptyset$, we fold $e$ into $T(e)$ and modify $\Psi$ as follows.
		For each $v \in T(e)$ we subtract $d_e$ from $c'_v$ and set $\ell_v$ to be $x_v$.
		Afterwards we remove the edge $e$ from $\MC{E}$.

		Note that, in the folding process, we implicitly assign the edge $e$ to 
		%round up the assignment $h_{e,v}$ to $1$ for 
		some particular $v \in T(e)$ and round down the assignment to zero for the remaining.
		Since $0 < h_{e,v} = x_v$ and $$d_e\cdot h_{e,v} \le \sum_{e' \in \MC{E}[v]}d_{e'}\cdot h_{e',v} \le c'_v\cdot x_v$$ hold for any $v \in T(e)$, we have $d_e \le c'_v$. 
		Provided that vertex $v$ is to be rounded up later in the final stage, the validity of this process is thereby guaranteed.
		%

%		We remark that, although this step is not required for applying Theorem~\ref{thm-structural-mapping}, it is essential in obtaining a tight approximation as it separates the part that induces large rounding error from vertices with medium-sized multiplicities.
\end{itemize}

If any refinement is made,
%After $\Psi$ is refined, 
this process repeats with the refined tuple of parameters.
% until we have reached a tuple that gives an empty $\bigcup_{e \in \MC{E}}T(e)$.
% is empty.
%
Otherwise we proceed to the rounding stage.

\paragraph{Final rounding of the vertices.}
Let $\tilde{p} = (\tilde{\bm{x}}, \tilde{\bm{h}})$ denote the basic optimal solution computed before the algorithm enters the rounding stage and $\tilde{\Psi} = (\MC{E},\bm{\ell},\bm{c'})$ denote the corresponding parameter tuple.
%for $\Psi'$. \mong{polish from here}
%
%
%\smallskip
%
%
Let $\bm{h'}$ denote the rounded assignment function for the edges that are folded, i.e., for all $e\in E\setminus \MC{E}$, $h'_{e,v}$ is $1$ for some particular $v \in T(e)$ and zero for the remaining.

\smallskip

\noindent
The final rounding is done as follows.
For any $v \in V$ and $e \in E[v]$, define 
$$x^*_v := \CEIL{\vphantom{\bigcup} \tilde{x}_v}
\quad \text{and} \quad
h^*_{e,v} := \begin{cases}
\hspace{2pt} \tilde{h}_{e,v}, \enskip & \text{if $e \in \MC{E}$,} \\
\hspace{2pt} h'_{e,v}, & \text{otherwise.}
\end{cases}
$$
Then $(\bm{x}^*, \bm{h}^*)$ is output as the solution.
%

%

%

%

%%%
%%%

\section{Proof of Theorem~\ref{thm-tight-vc-hc}}
\label{sec-proof-thm-tight-vc-hc}

To prove Theorem~\ref{thm-tight-vc-hc}, we show that 
\begin{itemize}
	\item
		Algorithm Tight-VC-HC terminates in $O(\BIGC{E})$ rounds and outputs a feasible solution
		%, $(x^*,h^*)$, for LP(\ref{ILP_cdh}) with respect to 
		for LP($\Psi_0$), where $\Psi_0 = (E,\mathbf{0},\bm{c})$ is the initial tuple.
		
	\item
		The output solution $(\bm{x}^*, \bm{h}^*)$ is an $f$-approximation for VC-HC.
		%, where $f = \max_{e \in E} \BIGC{e}$.
\end{itemize}

Consider the first statement to be proved. 
Since in each round of the execution, either the algorithm stops or at least one edge is folded.
It follows that the algorithm must terminate in $O(|E|)$ rounds.
Therefore it remains to prove the following:
\begin{itemize}[\qquad\quad {\bf --} \enskip]
	\item
		After each round, the refined tuple is feasible,
		% for LP(\ref{ILP_cdh}), 
		i.e., the corresponding LP has a non-empty feasible region.
		%at least one feasible solution.
		%, and therefore the folding operations of Tight-VC-HC are valid.
		%
		We prove this statement in Lemma~\ref{lemma-validity-tuple-refinement}.
		
	\item
		The output solution $(\bm{x}^*,\bm{h}^*)$ is feasible for LP($\Psi_0$).
		%the initual tuple. 
		This is proved in Lemma~\ref{lemma-folding-validity}.
\end{itemize}
Then we establish the approximation guarantee of $(\bm{x}^*,\bm{h}^*)$ in Section~\ref{sec-tight-vc-hc-analysis}.
%Lemma~\ref{lemma-rounding-bound-overall}.
%

%\smallskip

%
\paragraph{Notations and basic properties.}
Let $k$ denote the number of rounds the algorithm iterates before entering the rounding stage.
For $1\le i \le k$, we use the following notations to denote the respective concepts we have in the $i^{th}$ iteration:

\begin{itemize}
	\item
		$\Psi^{(i)} = \BIGP{ \hspace{2pt} \MC{E}^{(i)}, \bm{\ell}^{(i)}, \bm{c}^{(i)} \hspace{2pt}}$: The parameter tuple algorithm Tight-VC-HC maintains when it enters the $i^{th}$ rounds.
Here we have $\Psi^{(1)} = \Psi_0$ as a dummy notation for the initial tuple.

	\item
		$p^{(i)} = \BIGP{\bm{x}^{(i)}, \bm{h}^{(i)}}$: The basic optimal solution computed for LP($\Psi^{(i)}$).

	\item
		$I^{(i)}$: The set of vertices with small fractional values, i.e., $I^{(i)} := \{ v \in V \hspace{2pt} \colon \hspace{2pt} 0 < x^{(i)}_v < \frac{1}{f} \hspace{2pt} \}$.
\end{itemize}

%

%\smallskip

%
\noindent
The following proposition 
%Proposition~\ref{prop-basic-properties-tight-vc-hc} 
states a sandwich property for $\ell^{(i)}$ which we will be using later.

\begin{prop} \label{prop-basic-properties-tight-vc-hc}
For any $1\le i < k$ and any $v \in V$, we have 
$\hspace{2pt} \ell^{(i)}_v \le \hspace{2pt} \ell^{(i+1)}_v \le x^{(i)}_v$. Furthermore, $\ell^{(i)}_v > 0$ implies that $\cfrac{1}{f} \hspace{2pt} \le \hspace{2pt} \ell^{(i)}_v \hspace{2pt} \le \hspace{2pt} 1$.
\end{prop}

\begin{proof}%[Proof of Proposition~\ref{prop-basic-properties-tight-vc-hc}]
For the first statment, it suffices to see that, in the algorithm, $\ell^{(i+1)}_v$ equals either $\ell^{(i)}_v$ or $x^{(i)}_v$.
In both cases it implies that $\ell^{(i)}_v \le \ell^{(i+1)}_v \le x^{(i)}_v$.

\smallskip

The second statement follows from the non-decreasing property of $\ell^{(i)}$ with respect to $i$ and the fact that $\ell^{(i)}_v$ is updated only when $v \notin I$ and $0 < x^{(i)}_v = h^{(i)}_{e,v}$ for some $e$, which implies that $\cfrac{1}{f} \hspace{2pt} \le \hspace{2pt} x^{(i)}_v \hspace{2pt} \le \hspace{2pt} 1$.
\end{proof}

\medskip

\subsection{Feasibility of Tight-VC-HC}
\label{sec-tight-vc-hc-feasibility}

For any $1\le i \le k$ and any edge subset $E' \subseteq E$, we define the extended assignment $\bm{h}^{(i)}|_{E'}$ of $\bm{h}^{(i)}$ with respect to the edge set $E'$ as follows.
For any $e \in E'$ and any $v \in e$, let
$$\big( h^{(i)}|_{E'} \big)_{e,v} := \begin{cases}
h^{(i)}_{e,v}, & \text{if $e \in \MC{E}^{(i)}$}, \\
0, & \text{otherwise.}
\end{cases}$$
%
%\mong{define $h^{(i)}|_{\MC{E}^{(i+1)}}$ ***}

%\smallskip
\noindent
The following lemma shows that, the feasible region of LP($\Psi^{(i+1)}$) is not empty for any $1 \le i < k$.
Therefore the computation for $p^{(i)}$ results in a valid solution for all $1\le i \le k$.

\begin{lemma} \label{lemma-validity-tuple-refinement}
$\big( \bm{x}^{(i)}, \bm{h}^{(i)}|_{\MC{E}^{(i+1)}} \big)$ is feasible for LP($\Psi^{(i+1)}$), for any $1\le i < k$.
\end{lemma}

\begin{proof}%[Proof of Lemma~\ref{lemma-validity-tuple-refinement}]
We show that, if $p^{(i)}$ is feasible for LP($\Psi^{(i)}$), then the solution $\big( \bm{x}^{(i)}, \bm{h}^{(i)}|_{\MC{E}^{(i+1)}} \big)$ will be feasible for LP($\Psi^{(i+1)}$).
Note that, the validity of the base case, $i = 1$, follows from the assumption that the input instance is feasible.

\smallskip

\noindent
In the following, we show that $\big( \bm{x}^{(i)}, \bm{h}^{(i)}|_{\MC{E}^{(i+1)}} \big)$ does not violate the constraints in LP($\Psi^{(i+1)}$).
\begin{itemize}
	\item
		Since $\MC{E}^{(i+1)} \subseteq \MC{E}^{(i)}$, it follows that $\big( h^{(i)}|_{\MC{E}^{(i+1)}} \big)_{e,v} = h^{(i)}_{e,v}$ for all $e \in \MC{E}^{(i+1)}$ and all $v \in e$.
		Therefore the constraints~(\ref{LP_cdh_e}) and~(\ref{LP_cdh_ev}) hold directly for all $e \in \MC{E}^{(i+1)}$ and $v \in e$.

	\item
		Consider the constraint~(\ref{LP_cdh_mv}).
		We have $x^{(i)}_v \le m_v$ since $p^{(i)}$ is feasible for LP($\Psi^{(i)}$).
		By Proposition~\ref{prop-basic-properties-tight-vc-hc} we know that $\ell^{(i+1)}_v \le x^{(i)}_v$.
		Hence the constraint~(\ref{LP_cdh_mv}) holds for all $v \in V$.
		
	\item
		It remains to verify that the constraint~(\ref{LP_cdh_v}).
		For any $v \in V$, let $t^{(i)}(v)$ denote the set of edges in $\MC{E}^{(i)} \setminus \MC{E}^{(i+1)}$ that support $v$ in $p^{(i)}$.
		It follows that $$c^{(i)}_v - \sum_{e \in t^{(i)}(v)} d_e \enskip \le \enskip c^{(i+1)}_v.$$
		Expanding constraint~(\ref{LP_cdh_v}) for $v$, we have
		\begin{align}
		\sum_{e \in \MC{E}^{(i+1)}[v]}d_e \cdot {\big( h^{(i)}|_{\MC{E}^{(i+1)}} \big)}_{e,v}
		\enskip & = \enskip \sum_{e \in \MC{E}^{(i)}[v]}d_e \cdot h^{(i)}_{e,v} \enskip - \enskip \sum_{e \in \MC{E}^{(i)}[v] \setminus \MC{E}^{(i+1)}[v]}d_e \cdot h^{(i)}_{e,v}  \notag \\[4pt]
		& \le \enskip \sum_{e \in \MC{E}^{(i)}[v]}d_e \cdot h^{(i)}_{e,v} \enskip - \enskip \sum_{e \in t^{(i)}(v)}d_e \cdot h^{(i)}_{e,v}  \notag \\[2pt]
		& \le \enskip \Bigg( c^{(i)}_v - \sum_{e \in t^{(i)}(v)}d_e \Bigg) \cdot x^{(i)}_v  \notag \\[4pt]
		& \le \enskip c^{(i+1)}_v \cdot x^{(i)}_v. \notag
		\end{align}
		%
		%where the first equality follows from the definition of $t^{(i)}(v)$, and the second inequality follows from the feasibility of $\big( x^{(i)}, h^{(i)} \big)$ on $\Psi^{(i)}$ and the fact that $h^{(i)}_{e,v} = x^{(i)}_v$ for all $e \in t^{(i)}(v)$.
		%, and the last equality follows from the definition of $c^{(i+1)}}(v)$.
		Therefore the constraint~(\ref{LP_cdh_v}) holds for all $v \in V$ and this lemma is proved.
\end{itemize}
\vspace{-22pt}
\end{proof}

\medskip

Let $\bm{h}'$ denote the assignment function defined in the algorithm for the edges that have been folded.
The following lemma shows that the output solution $(\bm{x}^*, \bm{h}^*)$ is feasible for LP($\Psi_0$).

\begin{lemma} \label{lemma-folding-validity}
%
%For any $1\le i < k$ and any $p = (x,h)$ feasible for LP(\ref{ILP_cdh}) with respect to $\Psi^{(i+1)}$,
%
The solution $p = \big( \vphantom{\dfrac{}{}} \hspace{2pt} \CEIL{ \bm{x}^{(k)} }, \hspace{2pt} \bm{h}^{(k)}|_E + \bm{h}'|_E \hspace{2pt} \big)$ is feasible for LP($\Psi_0$).
\end{lemma}

\begin{proof}%[Proof of Lemma~\ref{lemma-folding-validity}]
We show that $p$ does not violate the constraints in LP($\Psi_0$).
\begin{enumerate}
	\item
		For the constraint~(\ref{LP_cdh_e}), it suffices to see that: (i)~For any $e \in \MC{E}^{(k)}$, $(h'|_E)_{e,v} = 0$ for all $v \in e$.
		Therefore $$\sum_{v \in e}\big( h^{(k)}|_E + h'|_E \big)_{e,v} \enskip = \enskip \sum_{v \in e}h^{(k)}_{e,v} \enskip = \enskip 1.$$
		%since $h^{(k)}$ is feasible for $\Psi^{(k)}$.
		%
		(ii)~For any $e \in E \setminus \MC{E}^{(k)}$, we know that $(h^{(k)}|_E)_{e,v} = 0$ for all $v \in e$ and $\sum_{v \in e}h'_{e,v} = 1$ by the definition of $h'$.
		% $T^{(i)}(e) \neq \emptyset$ for some $1\le i < k$, and therefore $\sum_{v \in e}h'_{e,v} = 1$ by the way $h'$ is defined.
		%
		Hence constraint~(\ref{LP_cdh_e}) holds for all $e \in E$.
		
	\item
		Consider the constraint~(\ref{LP_cdh_v}) for any $v \in V$.
		%
		%Depending on whether $x^{(k)}_v = 0$, w
		We consider two cases.

		\smallskip

		(i) If $x^{(k)}_v = 0$, then we have $h^{(k)}_{e,v} = 0$ for all $e \in \MC{E}^{(k)}$.

		Since $\ell^{(1)}_v = 0$, by Proposition~\ref{prop-basic-properties-tight-vc-hc} we have $\ell^{(i)}_v = 0$ for all $1\le i \le k$.
		This means that $v$ is not supported by any edge in $p^{(i)}$ for all $1\le i \le k$.
		% and all $e \in \MC{E}^{(i)}$.
		%
		Hence $(h'|_E)_{e,v} = 0$ for all $e \in E[v]$.
		Therefore, 
		$$\sum_{e \in E[v]}d_e \cdot \big( h^{(k)}|_E + h'|_E \big)_{e,v} \enskip = \enskip 0 \enskip = \enskip c_v \cdot \CEIL{x^{(k)}_v},$$
		and the constraint~(\ref{LP_cdh_v}) holds for $v$.

		\smallskip

		(ii) Consider the case $x^{(k)}_v > 0$.
		We have
		\begin{align*}
		\sum_{e \in E[v]}d_e \cdot \big( h^{(k)}|_E + h'|_E \big)_{e,v} 
		\enskip & = \enskip \sum_{e \in \MC{E}^{(k)}[v]}d_e \cdot h^{(k)}_{e,v} \enskip + \enskip \sum_{e \in E[v] \setminus \MC{E}^{(k)}}d_e\cdot h'_{e,v} \\[4pt]
		& \le \enskip c^{(k)}_v \cdot x^{(k)}_v \enskip + \enskip \sum_{e \in E[v] \setminus \MC{E}^{(k)} \colon h'_{e,v} = 1}d_e \\[1pt]
		& \le \enskip \Bigg( c^{(k)}_v + \sum_{e \in E[v] \setminus \MC{E}^{(k)} \colon h'_{e,v} = 1}d_e \Bigg) \cdot \CEIL{x^{(k)}_v} \enskip  \le \enskip c_v \cdot \CEIL{x^{(k)}_v}.
		\end{align*}
		%
		%where the second inequality uses the fact that $\big( x^{(k)}, h^{(k)} \big)$ is feasible for LP($\Psi^{(k)}$) and the fact that $h'_{e,v} \le 1$ for all $e \in t(v)$.
		%
		Therefore the constraint~(\ref{LP_cdh_v}) holds for $v$ as well.

	\item
		Consider the constraint~(\ref{LP_cdh_mv}) for any $v \in V$.
		By Proposition~\ref{prop-basic-properties-tight-vc-hc}, we know that $\ell^{(i)}_v$ is non-decreasing on $i$.
		Therefore it follows that 
		$$\CEIL{x^{(k)}_v} \enskip \ge \enskip x^{(k)}_v \enskip \ge \enskip \ell^{(k)}_v \enskip \ge \enskip \ell^{(1)}_v \enskip = \enskip 0.$$
		On the other hand, since $m_v$ is integral and since $x^{(k)}_v \le m_v$, it follows that $\CEIL{x^{(k)}_v} \le m_v$ as well.
		Hence constraint~(\ref{LP_cdh_mv}) holds for all $v \in V$.

	\item
		Consider the constraint~(\ref{LP_cdh_ev}) and any $e \in E$, $v \in e$.

		If $e \in \MC{E}^{(k)}$, then we have $\big( h^{(k)}|_E + h'|_E \big)_{e,v} = \hspace{2pt} h^{(k)}_{e,v} \hspace{2pt} \le \hspace{2pt} x^{(k)}_v \hspace{2pt} \le \hspace{2pt} \CEIL{x^{(k)}_v}$ and the constraint~(\ref{LP_cdh_ev}) holds naturally for $e,v$.

		If $e \notin \MC{E}^{(k)}$, then $\big( h^{(k)}|_E + h'|_E \big)_{e,v} = h'_{e,v}$. 
		%
%		If $h'_{e,v} = 0$, then $h'_{e,v} \le \CEIL{x^{(k)}_v}$ holds since $x^{(k)}_v \ge 0$.
		%
		By the definition of $h'$, we know that $h'_{e,v} > 0$ implies that $0 < h^{(i)}_{e,v} = x^{(i)}_v$ and $\ell^{(i)}_v = x^{(i)}_v$ for some $1 \le i < k$. 
		Therefore $$h'_{e,v} \enskip = \enskip 1 \enskip \le \enskip \CEIL{x^{(i)}_v} \enskip = \enskip \Big\lceil \ell^{(i)}_v \Big\rceil \enskip \le \enskip \Big\lceil \ell^{(k)}_v \Big\rceil \enskip \le \enskip \CEIL{x^{(k)}_v}.$$
		%
		%Since $h'_{e,v} \le 1$ for all $v \in e$, it follows that $h'_{e,v} \le \CEIL{x^{(k)}_v}$.
		%
		On the other hand, if $h'_{e,v} = 0$, then constraint~(\ref{LP_cdh_ev}) holds trivially.
		In both cases, we know that it holds for $e$ and $v$.
\end{enumerate}
%
%This proves the lemma.
\vspace{-21pt}
\end{proof}

\medskip

%%

%\paragraph{Approximation Guarantee.}
\subsection{Approximation Guarantee}
\label{sec-tight-vc-hc-analysis}

In this section we show that $(\bm{x}^*, \bm{h}^*)$ gives an $f$-approximation.
Since $x^*_v = \CEIL{x^{(k)}_v}$ for all $v \in V$, we will prove in Lemma~\ref{lemma-rounding-bound-overall} that 
$$\sum_{v \in V} \CEIL{x^{(k)}_v} \enskip \le \enskip f\cdot \sum_{v \in V}x^{(1)}_v.$$

\smallskip

Consider the two sets $I^{(k)}$ and $D := \BIGBP{ \hspace{2pt} v \hspace{2pt} \colon \hspace{2pt} v \in V, \enskip 1 < x^{(k)}_v < m_v \hspace{2pt}}.$
By their definitions and Proposition~\ref{prop-basic-properties-tight-vc-hc}, we know that $I^{(k)}$ and $D$ are mutually disjoint, and they consist of only non-extremal vertices.
Furthermore, $I^{(k)}$ is not supporting and $D$ is not supported.
%

%\mong{remove $T^{(i)}(e)$ ********}
%
%Consider the solution $p^{(k)} = \BIGP{x^{(k)}, h^{(k)}}$, which is extreme in $\mathbf{Q}(\Psi^{(k)})$, and the set $I^{(k)}$, which is non-extremal by its definition and Proposition~\ref{prop-basic-properties-tight-vc-hc}.
%
%Since $T^{(k)}(e) = \emptyset$ for all $e \in \MC{E}^{(k)}$, we know that $I^{(k)}$ is not a supporting vertex set.
%
%Define the set $$D := \BIGBP{ \hspace{2pt} v \hspace{2pt} \colon \hspace{2pt} v \in V, \enskip 1 < x^{(k)}_v < m_v \hspace{2pt}}.$$
%
%It follows that $D$ is a non-extremal set that is not supported, since $h^{(k)}_{e,v} \le 1$ for all $e \in E[D]$ and $v \in e$.
%

\medskip

\noindent
Therefore, by Lemma~\ref{thm-structural-mapping}, there exists a mapping 
\begin{equation}
\Gamma \colon \enskip I^{(k)} \enskip 
\mapsto \enskip \Big( \MC{E}^{(k)} \Big)^{\OP{actv}}_{h^{(k)}} \Big[ I^{(k)} \Big] 
\hspace{2pt} \setminus \hspace{2pt} \Big( \MC{E}^{(k)} \Big)^{\OP{actv}}_{h^{(k)}}\Big[D\Big]
\label{eq-image-gamma-mapping}
\end{equation}
such that: (i) $v\in \Gamma(v)$ for all $v \in I^{(k)}$, \enskip (ii) for any $u,v \in I^{(k)}$ with $u\neq v$, we have
\begin{equation}
\BIGP{\vphantom{\bigcup} \Gamma(u) \setminus I^{(k)}} \hspace{2pt} \cap \hspace{2pt} \BIGP{\vphantom{\bigcup} \Gamma(v) \setminus I^{(k)}} \enskip = \enskip \emptyset.
\label{eq-rounding-bound-partition-of-U}
\end{equation}
See also Figure~\ref{figure-rounding-bound-partition} for an illustration.
Note that, since we have $x^{(k)}_v < 1/f$ for all $v \in I^{(k)}$,
%by the definition of $I^{(k)}$, 
it follows that $$\sum_{v \in e \cap I^{(k)}} h^{(k)}_{e,v} < 1, \quad \text{and therefore} \quad e^{\OP{actv}}_{h^{(k)}} \not\subseteq I^{(k)} \quad \text{for all $e \in \MC{E}^{(k)}$.}$$
Hence it follows that $\Gamma(v) \setminus I^{(k)} \neq \emptyset$ for all $v \in I^{(k)}$.
In the following we will use the vertices in $\Gamma(v) \setminus I^{(k)}$ to help absorb the rounding cost of $v$.
%
%We also remark that, although this property is not used explicitly in our proof, the same conclusion can be derived from our arguments.
%the argument we use in the proof.
%

\begin{figure*}[tp]
\centering
\includegraphics[scale=0.7]{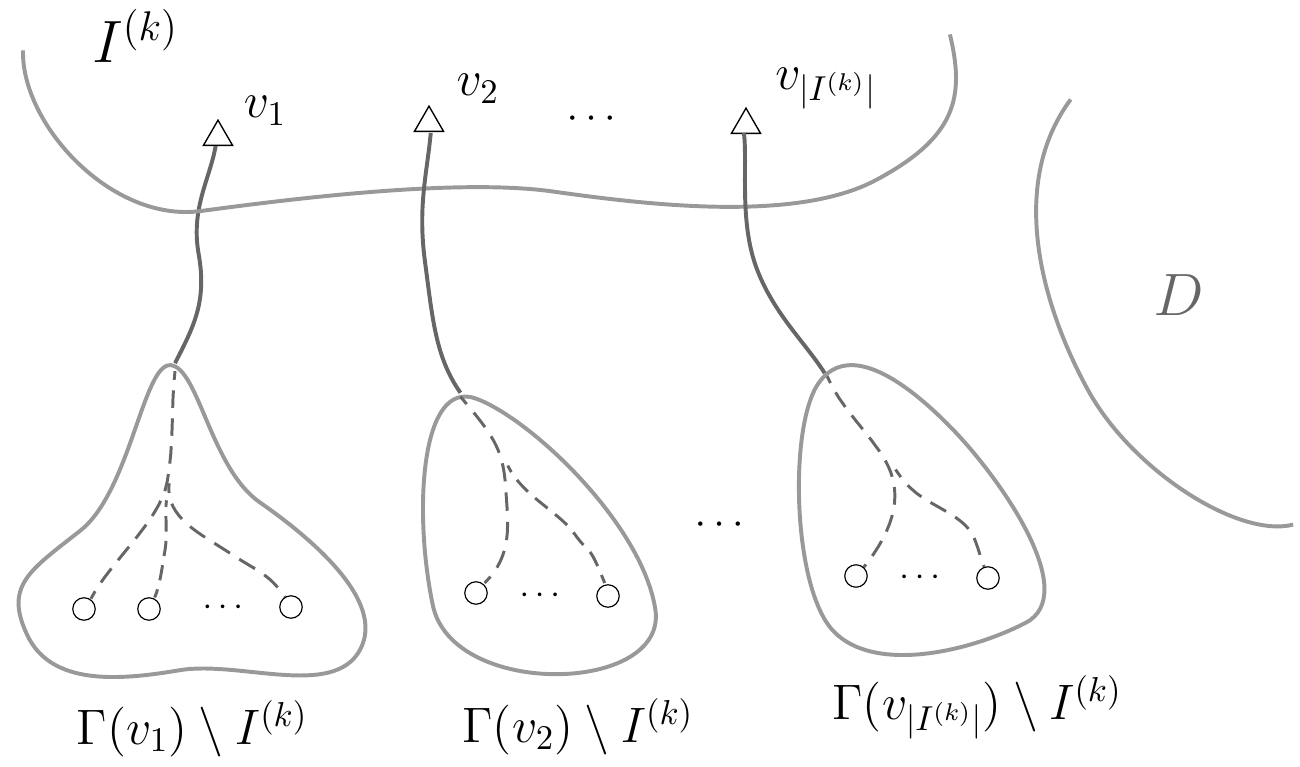}
\caption{An illustration on the partition given by $\Gamma$.}
\label{figure-rounding-bound-partition}
\end{figure*}

\medskip

For any $v \in I^{(k)}$, define $\pi(v) := \BIGP{ \vphantom{\dfrac{}{}} \hspace{2pt} \Gamma(v) \setminus I^{(k)} \hspace{2pt}} \cup \{\vphantom{\cup}v\}$.
%
%Note that, since $v \in \Gamma(v)$, it follows that $\pi(v) \subseteq \Gamma(v)$.
%
%\smallskip
%
%
The following lemma bounds the rounding cost for the vertices included by an individual $\pi(v)$.
% of vertices.

%

\begin{lemma} \label{lemma-rounding-bound-single-vertex}
Provided that $f \ge 2$, we have
$$\sum_{u \in \pi(v)}\CEIL{x^{(k)}_u} \enskip \le \enskip f\cdot \sum_{u \in \pi(v)}x^{(k)}_u, \quad \text{for any $v \in I^{(k)}$.}$$

\end{lemma}

\begin{proof}%[Proof of Lemma~\ref{lemma-rounding-bound-single-vertex}]
From the Condition~(\ref{eq-image-gamma-mapping}) above, we know that $\Gamma(v)$ contains no vertices in $D$.
% $\Gamma(v) \notin \big( \MC{E}^{(k)} \big)^{\OP{actv}}_{h^{(k)}}\big[D\big]$.
%
%Hence, by definition we know that $\Gamma(v) \cap D = \emptyset$.
%
Since $\pi(v) \subseteq \Gamma(v)$, it follows that $\pi(v) \cap D = \emptyset$ as well.
Depending on the elements of $\pi(v)$, we consider two cases.

\smallskip

\begin{itemize}
	\item
		If there exists some $u \in \pi(v)$ with $x^{(k)}_u = m_u$, then we have $$\CEIL{x^{(k)}_v} + \CEIL{x^{(k)}_u} \enskip = \enskip 1 + m_u \enskip \le \enskip f\cdot \Big( \hspace{2pt} x^{(k)}_v + x^{(k)}_u \hspace{2pt} \Big),$$
		since $f \ge 2$ and $m_u \ge 1$.
		For the remaining vertices $u' \in \pi(v) \setminus \{ u \cup v \}$, we know that $x^{(k)}_{u'} \ge 1/f$ since $u' \in \Gamma(v) \setminus I^{(k)}$ by definition of $\pi$.
		Hence $\vphantom{\cfrac{}{}} \CEIL{x^{(k)}_{u'}} \le f\cdot x^{(k)}_{u'}$, and the statement of this lemma holds.

	\smallskip
	
	\item
		Consider the other case that $x^{(k)}_u \neq m_u$ for all $u \in \pi(v)$.
		Then we know that $x^{(k)}_u \le 1$ for all $u \in \pi(v)$ and hence
		\begin{equation}\sum_{u \in \pi(v)}\CEIL{x^{(k)}_u} 
\enskip = \enskip \BIGC{\vphantom{\bigcup} \pi(v)}.
		\label{eq-rounding-bound-per-vertex-pi-size}
		\end{equation}
		In the following we bound $\BIGC{\vphantom{\bigcup} \pi(v)}$.
Let $e \in \MC{E}^{(k)}[v]$ be an edge such that $\Gamma(v) = e^{\OP{actv}}_{h^{(k)}}$.
%
%Note that, s
Such an edge exists since $\Gamma(v)$ is the active subedge of certain edge in $\MC{E}^{(k)}[I^{(k)}]$ which also contains $v$.
%
%Since $\pi(v) \subseteq \Gamma(v)$, i
It follows that $\pi(v) \subseteq \Gamma(v) \subseteq e$. 

\smallskip

We have
\begin{align}
\sum_{u \in \pi(v)}x^{(k)}_u 
\enskip \ge \enskip \sum_{u \in \pi(v)}h^{(k)}_{e,u} 
\enskip = \enskip 1 - \sum_{u \in \Gamma(v) \setminus \pi(v)}h^{(k)}_{e,u} \label{eq-rounding-bound-per-vertex-edge-feasibility} 
\end{align}
where the first inequality follows from constraint~(\ref{LP_cdh_ev}) and the second equality follows from constraint~(\ref{LP_cdh_e}) for $e$.
% and the fact that $\Gamma(v) = e^{\OP{actv}}_{h^{(k)}}$.
%

\smallskip

Since $\Gamma(v) \setminus \pi(v) \subseteq I^{(k)}$ by the definition of $\pi(v)$,
% and since we have $h^{(k)}_{e,u} \le x_u < \frac{1}{f}$ for all $u \in e \cap I^{(k)}$, 
it follows that
\begin{align}
1 - \sum_{u \in \Gamma(v) \setminus \pi(v)}h^{(k)}_{e,u} 
\enskip & > \enskip 1 - \frac{1}{f}\cdot \hspace{2pt} \BIGC{\vphantom{\bigcup} \hspace{2pt} \Gamma(v) \setminus \pi(v) \hspace{2pt}} \hspace{4pt}
\enskip \ge \enskip \frac{1}{f} \cdot \BIGC{\vphantom{\bigcup} \pi(v)},
\label{ieq-rounding-bound-per-vertex-size-Uv}
\end{align}
where in the second inequality we use the fact that
%$\BIGC{\Gamma(v)} \in \MC{E}^{(k),\OP{actv}}_{h^{(k)}}$, we know that 
$\BIGC{\Gamma(v)} \le \BIGC{e} \le f$, which implies that $1 - \frac{1}{f}\cdot \BIGP{ \hspace{2pt} \BIGC{\Gamma(v)} - \BIGC{\pi(v)} \hspace{2pt} }\ge \frac{1}{f}\cdot\BIGC{\pi(v)}$.
%, it follows that
%
%Therefore,
%\begin{equation}
%1 - \frac{1}{f}\BIGP{ \hspace{4pt} \BIGC{\vphantom{\bigcup}\Gamma(v)} -  \BIGC{\vphantom{\bigcup}\pi(v)} \hspace{4pt}}
%\enskip \ge \enskip \frac{1}{f}\BIGC{\vphantom{\bigcup} \pi(v)}.
%\label{ieq-rounding-bound-per-vertex-size-sigma-v}
%\end{equation}
%
%
%
Combining~(\ref{eq-rounding-bound-per-vertex-pi-size}),~(\ref{eq-rounding-bound-per-vertex-edge-feasibility}), and~(\ref{ieq-rounding-bound-per-vertex-size-Uv}), we obtain
$$\sum_{u \in \pi(v)}\CEIL{x^{(k)}_u} 
\enskip = \enskip \BIGC{\vphantom{\bigcup} \pi(v)} 
\enskip \le \enskip f\cdot \sum_{u \in \pi(v)}x^{(k)}_u$$
as claimed.
\end{itemize}
\vspace{-22pt}
\end{proof}

\noindent
The following lemma establishes the approximation guarantee for the solution $(x^*, h^*)$.

\begin{lemma} \label{lemma-rounding-bound-overall}
Provided that $f \ge 2$, we have
$$\sum_{u \in V}\CEIL{x^{(k)}_u} \enskip \le \enskip f\cdot \sum_{u \in V}x^{(1)}_u.$$
\end{lemma}

\begin{proof}%[Proof of Lemma~\ref{lemma-rounding-bound-overall}]
From
%By Condition~(\ref{eq-rounding-bound-partition-of-U}) and 
the definition of $\pi$, we know that $\pi(u) \cap \pi(v) = \emptyset$ for any $u,v \in I^{(k)}$ with $u \neq v$.
%
%Let $U := V \setminus \bigcup_{v \in I^{(k)}}\pi(v)$.
%
%Then it follows that $\pi$ and $U$ gives a partition for the vertex set $V$.
%
%
%\smallskip
%
%
%Consider any $u \in U$.
%
%Since we have $I^{(k)} \subseteq \bigcup_{v \in I^{(k)}}\pi(v)$ by the definition of $\pi$, from the definition of $U$ we know that $u \notin I^{(k)}$.
%
Furthermore, for any $u \in V \setminus \bigcup_{v \in I^{(k)}}\pi(v)$, we have 
%This means that 
$x^{(k)}_u \ge {1}/{f}$ whenever $x^{(k)}_u > 0$.

\smallskip

Therefore, combining with Lemma~\ref{lemma-rounding-bound-single-vertex}, we have
\begin{align*}
\sum_{u \in V} \CEIL{x^{(k)}_u} 
\enskip & = \enskip \sum_{v \in I^{(k)}} \sum_{u \in \pi(v)} \CEIL{x^{(k)}_u} \quad + \sum_{u \in V \setminus \bigcup_{v \in I^{(k)}}\pi(v)} \CEIL{x^{(k)}_u} \\[2pt]
& \le \enskip \sum_{v \in I^{(k)}} f\cdot \sum_{u \in \pi(v)} x^{(k)}_u \enskip + \enskip f\cdot \sum_{u \in V \setminus \bigcup_{v \in I^{(k)}}\pi(v)} x^{(k)}_u \\[2pt]
& \le \enskip f\cdot \sum_{u \in V} x^{(k)}_u.
\end{align*}
By Lemma~\ref{lemma-validity-tuple-refinement}, $\vphantom{\cfrac{}{}} \sum_{u \in V}x^{(i+1)}_u \le \sum_{u \in V}x^{(i)}_u$ holds for all $1\le i < k$, since $x^{(i)}$ is feasible for LP($\Psi^{(i+1)}$) while $x^{(i+1)}$ is optimal for LP($\Psi^{(i+1)}$).
Therefore $$\sum_{u \in V}\CEIL{x^{(k)}_u} \enskip \le \enskip f\cdot \sum_{u \in V}x^{(k)}_u \enskip \le \enskip f\cdot \sum_{u \in V}x^{(1)}_u$$
as claimed.
\end{proof}

\newpage

\section{Proof of Lemma~\ref{thm-structural-mapping}}
\label{sec-proof-separation-lemma}

In this section we prove our main technical tool Lemma~\ref{thm-structural-mapping}.
We remark that, throughout this proof, we will assume the following prerequisite (as stated in the lemma):
\begin{enumerate}
	\item
		$\Psi = (E, \bm{\ell}, \bm{c})$ denotes the considered parameter tuple,
		
	\item
		$p = (\bm{x}, \bm{h})$ is an extreme point solution for $\mathbf{Q}(\Psi)$, and

	\item
		$\MC{I}$ and $\MC{D}$ are two disjoint sets of non-extremal vertices (with reference to the point $p$) such that $\MC{I}$ is not supporting and $\MC{D}$ is not supported.
\end{enumerate}
With the prerequisite above, the proof proceeds as follows.
First in Section~\ref{sec-separation-constraint-matrix} we identify a matrix $\tilde{M}$ from the constraint matrix of $\mathbf{Q}(\Psi)$ according to $\MC{I}$ and $\MC{D}$ together with the set of constraints they have involved in.
%
%By properly reducing the constraints, a further submatrix $\tilde{M}$ with a full column rank is obtained in Section~\ref{sec-separation-reduced-submatrix}.
%
From the matrix $\tilde{M}$ we show in Section~\ref{sec-separation-mapping-from-extremality} that a mapping with nice and restricted behavior can be extracted for the final mapping $\Gamma$ to be defined in Section~\ref{sec-separation-mapping-gamma}.
%

%

%\paragraph{The constraint matrix of LP(\ref{ILP_cdh}) and the notations.}
\subsection{The constraint matrix of $\mathbf{Q}(\Psi)$ and the matrix $\tilde{M}$}
\label{sec-separation-constraint-matrix}
Let $M$ denote the coefficient matrix for the constraints 
%constraint matrix (the coefficient matrix) of
in $\mathbf{Q}(\Psi)$ and $M_{(\OP{Eqs})}$ denote the submatrix formed 
%only 
by the constraints that hold with equality at the considered extreme point $p$.
%
%
%\smallskip
%
%
The following proposition states an equivalent form 
%for
of a classical characterization on the extremality of $p$ which we will be using throughout the proof. 
%
%\mong{more words for this ***}
%

\begin{prop} \label{claim-MID-full-column-rank}
$M_{(Eqs)}$ has a rank equal to the number of variables in $\mathbf{Q}(\Psi)$.
\end{prop}

%\begin{proof}%[Proof of Claim~\ref{claim-MID-full-column-rank}]
%
%\mong{***}
%
%\end{proof}

%
In the following, we first identify the sets of variables and constraints that are relative to $\MC{I}$ and $\MC{D}$.
Then we show that, by properly simplifying the constraints, a matrix $\tilde{M}$ with a full column rank and nice structural property can be identified.
%

%\smallskip

%
\paragraph{Relative variables and constraints.}
\noindent
Consider two sets $X$ and $H$ of variables defined as 
$$X := \BIGBP{ \enskip x_v \hspace{2pt} \colon \hspace{2pt} v \in \MC{I} \cup \MC{D} \enskip } 
\quad \text{and} \quad 
H := \BIGBP{\enskip h_{e,v} \hspace{2pt} \colon \hspace{2pt} e \in E[\MC{I}] \enskip \text{such that} \enskip v \in {e\hspace{1pt}}^{\OP{actv}}_h \enskip }.$$
Let $\OP{Eqs}(X, H)$ denote the set of constraints in which the variables in $X \cup H$ have involved and which also hold with equality.
% at the considered extreme point $p$.
%
%
%\smallskip
%
%
%\noindent
We classify the equalities of $\OP{Eqs}(X,H)$ by considering the four categories of constraints in $\mathbf{Q}(\Psi)$ as follows:
\begin{itemize}
	\item
		For the constraint~(\ref{LP_cdh_e}), we use
		\[
		\MC{C}^{(E)} := \BIGBP{ \vphantom{\dfrac{}{}}\hspace{3pt} e \hspace{2pt} \colon \hspace{2pt} \fbox{$\vphantom{{\dfrac{}{}}^{\bigcup}}$ $\sum_{v\in e}h_{e,v} = 1$ } \in \OP{Eqs}(X,H) \hspace{3pt} }
		\]
		to denote the the set of edges whose constraints are in $\OP{Eqs}(X,H)$.
		Note that, it follows that $\MC{C}^{(E)} = E[\MC{I}]$, since for any $e \in E[\MC{I}]$, at least one variable of the constraint $\sum_{v\in e}h_{e,v} = 1$ is included in $H$ by definition.
		%if and only if $e \in E[\MC{I}]$.
		%

	\item
		For the constraint~(\ref{LP_cdh_v}), we use 
		\[
			\MC{C}^{(V)} := \BIGBP{\hspace{3pt} v \hspace{2pt} \colon \hspace{2pt} \fbox{$\vphantom{{\dfrac{}{}}^{\bigcup}}$ $\sum_{e \in E[v]}d_e \cdot h_{e,v} \hspace{2pt} \le \hspace{2pt} c_v \cdot x_v$ } \in \OP{Eqs}(X,H) \hspace{3pt}}
		\]
		to denote the set of vertices whose capacity constraints are in $\OP{Eqs}(X,H)$.
		%
		%Note that $\MC{C}^{(V)} \subseteq \BIGP{\MC{I} \cup \MC{D}}$.

	\item
		Since $\MC{I}$ and $\MC{D}$ are non-extremal, it follows that $\ell_v < x_v < m_v$ for all $x_v \in X$.
		Therefore the constraint~(\ref{LP_cdh_mv}): $\ell_v \le x_v \le m_v$ does not appear in $\OP{Eqs}(X,H)$ for all $v \in \MC{I} \cup \MC{D}$.

	\item
		For the constraint~(\ref{LP_cdh_ev}), since $h_{e,v} > 0$ for all $h_{e,v} \in H$, the constraint $0\le h_{e,v}$ does not appear in $\OP{Eqs}(X,H)$ for all $e \in E$ and $v \in e$.
		On the other hand, we use 
		\[
			\MC{C}^{(E \times V)} := \BIGBP{ \vphantom{\dfrac{}{}} \enskip (e,v) \hspace{2pt} \colon \hspace{2pt} x_v \in X, \enskip h_{e,v} \in H, \enskip \fbox{$\vphantom{{\frac{}{}}^{\cup}}$ $h_{e,v} \hspace{1pt} \le \hspace{2pt} x_v$ } \in \OP{Eqs}(X,H) \enskip}
		\]
		to denote the set of supporting constraints in $\OP{Eqs}(X,H)$.
		Since $\MC{I}$ is not supporting and $\MC{D}$ is not supported, we know that $\vphantom{{\dfrac{}{}}^{\frac{T}{T}}} (e,v) \in \MC{C}^{(E \times V)}$ implies that $v \in \MC{I}$.

\end{itemize}
Note that, from the classification above, 
it follows that $\OP{Eqs}(X,H) = \MC{C}^{(E)} \cup \MC{C}^{(V)} \cup \MC{C}^{(E \times V)}$.
%, and $\MC{C}_{\OP{ieq}}$.
%
Let $$H^* := H \setminus \BIGBP{\vphantom{\dfrac{}{}} \hspace{2pt} h_{e,v} \hspace{2pt} \colon \hspace{2pt} (e,v) \in \MC{C}^{(E \times V)} \hspace{2pt}}$$ denote the set of variables in $H$ that does not correspond to a supporting constraint.
% in $\MC{C}^{(E \times V)}$.
%
For the simplicity of notations, we will also use $H^*$ to denote the set of pairs $(e,v)$ such that $h_{e,v}$ is contained in $H^*$, when there is no confusion in the context.
%

%\smallskip

%%
%%

%

\paragraph{The coefficient matrix $M_{(X,H)}$.}

Let $M_{(X,H)}$ denote the submatrix of $M_{(\OP{Eqs})}$ that is formed by $X \cup H$ and $\OP{Eqs}(X,H)$.
\vspace{-12pt}
\[
	\setlength{\arraycolsep}{0pt}
	\setlength{\dashlinedash}{2pt}
	\setlength{\dashlinegap}{1pt}
	\newcolumntype{C}[1]{>{\centering\arraybackslash}m{#1}}
	\newcolumntype{\Empty}{@{}m{0pt}@{}}
	\begin{array}{C{5em}c@{}c}
		&	\begin{array}{ C{3.6em} @{} C{5.8em} \Empty}
				& $\overbrace{\hspace{5.4em}}^{\text{\normalsize $X, H$}}$ & 
			\end{array} \\[-6pt]
	$M_{(\OP{Eqs})}$ = &
		\left(
		\begin{array}{ C{3.6em} : C{5.8em} \Empty}
			$\ddots$ & $M_{(X,H)}$ & \\[1.6em] \cdashline{1-2}
			$\cdots$ & $\mathbf{0}$ & \\[1em]
		\end{array}
		\right) &
		\hspace{-4pt}
		\begin{array}{ C{1em}@{} c \Empty}
		$\left.\begin{minipage}[c][3em]{1em}\end{minipage}\right\}$ & \rotatebox[origin=c]{0}{ $\OP{Eqs}(X, H)$} & \\
		\\[1em]
		\end{array}
	\end{array}
\]
%

%
%\paragraph{The matrix $M_{(X,H)}$ and the reduced submatrix $\tilde{M}$.}
%\subsection{The matrix $M_{(X,H)}$ and the reduced submatrix $\tilde{M}$}
%\label{sec-separation-reduced-submatrix}
%
Provided the above classification for constraints in $\OP{Eqs}(X,H)$ and the definition of $H^*$, the matrix $M_{(X,H)}$ can be written as follows, where the submatrices $M^{(1)}_{\MC{C}^{(E \times V)}}$ and $M^{(2)}_{\MC{C}^{(E \times V)}}$ are described in Figure~\ref{fig-matrix-C-EV}.

\smallskip

\[
	\newcommand*{\tn}[3]{\tikz[remember picture]\node[#1] (#2) {#3};}
	\setlength{\dashlinegap}{1pt}
	\setlength{\dashlinedash}{2pt}
	\newcolumntype{C}[1]{>{\centering\arraybackslash}m{#1}}
	\newcolumntype{\Empty}{@{}m{0pt}@{}}
	\begin{array}{cc@{}c}
		&	\begin{array}{ C{5em} C{3.8em} C{4em} C{3.4em} \Empty}
				$h_{e,v}$, & $h_{e,v}$, & $x_v$, & $x_v$, \\
				{\footnotesize \hspace{-6pt}$(e,v) \in \parbox{2.6em}{$\MC{C}^{(E \times V)}$}$} & 
				{\footnotesize $(e,v) \in \parbox{1em}{$H^*$}$} & 
				{\footnotesize $v \in \MC{I}$} &
				{\footnotesize $v \in \MC{D}$}
			\end{array} \\
		M_{(X,H)} = 
		&
		\left(
		\begin{array}{ C{5em} : C{3.8em} : C{4em} : C{3.4em} \Empty}
			\multirow{2}*[-6pt]{ \tn{anchor=east,inner sep=5pt}{f1}{$\vdots$} } & 
			\multirow{3}*[-12pt]{ $\vdots$ } & 
			\tn{anchor=east,inner sep=4pt}{f2}{$\mathbf{0}$} & $\mathbf{0}$ &
			\\[1.4em] 
			\cdashline{3-4}
			& & \tn{anchor=east,inner sep=4pt}{f3}{$\vdots$} & $\mathbf{0}$ &
			\\[1em] 
			\cdashline{1-1}
			\cdashline{3-4}
			$\mathbf{0}$ & & $\mathbf{0}$ & \hspace{2pt}$M^{\MC{D}}_{\MC{C}^{(V)}}$ & \\[1.4em]
			\cdashline{1-4}
			\tn{anchor=east,inner sep=2pt}{s1}{$M^{(1)}_{\MC{C}^{(E \times V)}}$} & 
			$\mathbf{0}$ & 
			\hspace{4pt}\tn{anchor=east,inner sep=0pt}{s2}{$M^{(2)}_{\MC{C}^{(E \times V)}}$} & $\mathbf{0}$ &
			\\[1.6em]
		\end{array}
		\hspace{2pt}\right)
		&
		\begin{array}{ m{12em} @{} \Empty }
			$\MC{C}^{(E)} \colon \sum h_{e,v} = 1$ & \\[1.2em]
			\cdashline{1-1}
			\multirow{2}*[-8pt]{ \hspace{-8pt} $\MC{C}^{(V)} \colon \sum d_e\cdot h_{e,v} \le c_v\cdot x_v$} & \\[1.3em]
			& \\[1.4em]
			\cdashline{1-1}
			$\MC{C}^{(E \times V)} \colon \enskip h_{e,v} \le x_v$ & \\[1.6em] 
		\end{array}
	\end{array}
	\tikzstyle{cvd_edge} = [black!80,dashed,->,shorten <= 2pt, shorten >= 2pt]
	\tikz[remember picture,overlay]
	\path[->] (s1.east)[yshift=0.2cm,xshift=-0.4cm] edge[cvd_edge,out=45,in=345] (f1.south east);
	\tikz[remember picture,overlay]
	\path[->] (s2.east)[xshift=-0.4cm,yshift=0.2cm] edge[cvd_edge,out=60,in=325] (f2.east);
	\tikz[remember picture,overlay]
	\path[->] (s2.east)[xshift=-0.4cm,yshift=0.2cm] edge[cvd_edge,out=60,in=340] (f3.east);
\]
\smallskip

\begin{figure*}[tp]
\[
	\setlength{\dashlinegap}{1pt}
	\setlength{\dashlinedash}{1pt}
	\newcolumntype{C}[1]{>{\centering\arraybackslash}m{#1}}
	\newcolumntype{\Empty}{@{}m{0pt}@{}}
	\begin{array}{c@{}c@{}c}
		&	{
			\small
			\setlength{\arraycolsep}{0pt}
			\begin{array}{ *{16}{C{1.9em}} \Empty}
				& & & & & &
				$x_{v_1}$ & $x_{v_2}$ & {\footnotesize $\cdots$} & $x_{v_k}$ & {\footnotesize $\cdots$} & $x_{v_{|\MC{I}|}}$
			\end{array} 
			} \hspace{8pt} \\[4pt]
		\left[\begin{array}{c|c@{}}
			\\[-6pt]
			M^{(1)}_{\MC{C}^{(E \times V)}} & M^{(2)}_{\MC{C}^{(E \times V)}} \\[8pt]
		\end{array}\right] \enskip =
		&
		\enskip {
		\setlength{\arraycolsep}{0pt}
		\left(
		\begin{array}{ *{9}{C{1em}} | *{4}{C{1em}} *{3}{C{1em}} \Empty}
			\cdashline{1-3}
%% v1
			\multicolumn{3}{:c:}{ \multirow{3}*{ $I_{v_1}$ } } & & & 
				\multicolumn{4}{c|}{ \multirow{5}*{ $\mathbf{0}$ \enskip } } & 
					\multicolumn{1}{c:}{
						\multirow{3}*{ \begin{minipage}{1em}
						\footnotesize\centering
						$\begin{array}{@{}c@{}}
							-1 \\[-2pt] -1 \\[-6pt] \cdot \\[-8pt] \cdot \\[-8pt] \cdot \\[-4pt] -1
						\end{array}$
						\end{minipage} } } & & 
					\multicolumn{2}{c:}{ \multirow{5}*{ $\mathbf{0}$ \enskip } }	&
					\multicolumn{3}{c}{ \multirow{9}*{ $\mathbf{0}$ } }  \\
			\multicolumn{3}{:c:}{} & & &
				\multicolumn{4}{c|}{} & 
					\multicolumn{1}{c:}{} & & & \multicolumn{1}{c:}{} \\
			\multicolumn{3}{:c:}{} & & &
				\multicolumn{4}{c|}{} & 
					\multicolumn{1}{c:}{} & & & \multicolumn{1}{c:}{} \\\cdashline{1-5}\cdashline{10-11}
%% v2
			& & & \multicolumn{2}{:c:}{ \multirow{2}*{ $I_{v_2}$ } } & 
				\multicolumn{4}{c|}{} & & 
					\multicolumn{1}{:c:}{
						\multirow{2}*{ \begin{minipage}{1em}
						\footnotesize\centering
						$\begin{array}{@{}c@{}}
							-1 \\[-6pt] \cdot \\[-8pt] \cdot \\[-8pt] \cdot \\[-4pt] -1
						\end{array}$
						\end{minipage} } } & & \multicolumn{1}{c:}{} & & & \\
			& & & \multicolumn{2}{:c:}{} & 
				\multicolumn{4}{c|}{} & &
					\multicolumn{1}{:c:}{} & & \multicolumn{1}{c:}{} \\\cdashline{4-5}\cdashline{11-11}
%% ...
			\multicolumn{5}{c}{ \multirow{4}*{ $\mathbf{0}$ } } &
				\multicolumn{2}{c}{ \multirow{2}*{ $\ddots$ } } & & &
					\multicolumn{2}{c}{ \multirow{4}*{ \quad $\mathbf{0}$ } } & 
					\multicolumn{1}{c}{ 
						\multirow{2}*{ $\ddots$ } } & \multicolumn{1}{c:}{} \\
			\multicolumn{5}{c}{} & 
				\multicolumn{2}{c}{} & & & & & & \multicolumn{1}{c:}{} \\\cdashline{8-9}\cdashline{13-13}
%% vk
			\multicolumn{5}{c}{} & & & 
				\multicolumn{2}{:c|}{ \multirow{2}*{ $I_{v_k}$ } } & & & & 
					\multicolumn{1}{:c:}{
						\multirow{2}*{ \begin{minipage}{1em}
						\footnotesize\centering
						$\begin{array}{@{}c@{}}
							-1 \\[-6pt] \cdot \\[-8pt] \cdot \\[-8pt] \cdot \\[-4pt] -1
						\end{array}$
						\end{minipage} } } \\
			\multicolumn{5}{c}{} & & & 
				\multicolumn{2}{:c|}{} & & & &
					\multicolumn{1}{:c:}{} \\\cdashline{8-9}\cdashline{13-13}
		\end{array}
		\hspace{2pt}\right)
		}
		&
		\begin{array}{ @{} C{6em} @{} \Empty }
			\multirow{3}*{ $(e,v_1) \in \MC{C}^{(E \times V)}$ } & \\
			\\
			\\\cdashline{1-1}
			\multirow{2}*{ $(e,v_2) \in \MC{C}^{(E \times V)}$ } & \\
			\\\cdashline{1-1}
			\multirow{2}*{ $\vdots$ } & \\
			\\\cdashline{1-1}
			\multirow{2}*{ $(e,v_k) \in \MC{C}^{(E \times V)}$ } & \\
			\\
		\end{array}
	\end{array}
\]
\vspace{-8pt}
\caption{The matrices $M^{(1)}_{\MC{C}^{(E \times V)}}$ and $M^{(2)}_{\MC{C}^{(E \times V)}}$, where $I_{v_i}$ denotes the identity matrix with dimension $\BIGC{\BIGBP{e\colon (e,v_i) \in \MC{C}^{(E \times V)}}}$.}
\label{fig-matrix-C-EV}
\end{figure*}

\smallskip

For the submatrix $\vphantom{\cfrac{}{}} M^{\MC{D}}_{\MC{C}^{(V)}}$, consider the column to which the variable $x_v$ corresponds for any $v \in \MC{D}$.
Since $M_{(\OP{Eqs})}$ has full column rank and since $x_v$ does not appear in the constraints of $\MC{C}^{(E)}$ and $\MC{C}^{(E \times V)}$,
% for all $v \in \MC{D}$, 
it follows that $v$ must be in $\MC{C}^{(V)}$.\footnote{Otherwise, the column $x_v$ would be a zero vector, a contradiction.}
Therefore the matrix $M^{\MC{D}}_{\MC{C}^{(V)}}$ can be written as a diagonal matrix with coefficients $\BIGBP{\vphantom{\frac{}{}}\hspace{-3pt} -c_v}_{v \in \MC{D}}$ on its diagonals. 
%

%\smallskip

%%
%%

\paragraph{Reduced constraints $\tilde{\MC{C}}^{(E)}$ and $\tilde{\MC{C}}^{(V)}$, and the submatrix $\tilde{M}$.}
%
%\mong{polish *******}
%
Use the identity elements in the diagonal of $\vphantom{{\dfrac{}{}}_{T}} M^{(1)}_{\MC{C}^{(E \times V)}}$ as pivots and perform row reduction (Gaussian elimination on the rows) on $M_{(X,H)}$.
Let the resulting matrix be $\vphantom{{\dfrac{}{}}_{T}} M'_{(X,H)}$.
Note that, in this process, we literally replace the variable $h_{e,v}$ with $x_v$ in the constraints $\MC{C}^{(E)}$ and $\MC{C}^{(V)}$ for all $\vphantom{{\dfrac{}{}}^{\frac{T}{T}}} (e,v) \in \MC{C}^{(E \times V)}$.
In particular, for each $e \in \MC{C}^{(E)}$, the constraint $\sum_{v \in e}h_{e,v} = 1$ is replaced by the constraint\footnote{Note that, in Equation~(\ref{cons-C-E-updated}) we also drop out $h_{e,v}$ for all $e \in \MC{C}^{(E)}$ and $v \in e \setminus e^{\OP{actv}}_h$ since by definition they are zero at the considered extreme point $p$.}
\begin{equation}
\sum_{\substack{v \text{ such that} \\ (e,v) \in H^*}}h_{e,v} \enskip + \sum_{\substack{v \text{ such that}, \\ (e,v) \in \MC{C}^{(E \times V)}}}x_v \hspace{2pt} = \hspace{2pt} 1.
\label{cons-C-E-updated}
\end{equation}
Similarly, for each $\vphantom{{\dfrac{}{}}^{\frac{T}{T}}} v \in \MC{C}^{(V)}$, the original constraint is replaced by
\begin{equation}
\vphantom{{\dfrac{}{}}^{{\bigcup}^{\bigcup}}}
\sum_{\substack{e \in E[v] \text{ such that} \\ (e,v) \notin \MC{C}^{(E \times V)}}}d_e \cdot h_{e,v} \hspace{2pt} \le \hspace{2pt} c'_v \cdot x_v,
\quad \text{where } \hspace{2pt}
c'_v = \hspace{2pt} c_v \hspace{2pt} - \hspace{-2pt} \sum_{\substack{ e \text{ such that}, \\ (e,v) \in \MC{C}^{(E \times V)}}}d_e.
\phantom{text{ttt}}
\label{cons-C-V-updated}
\end{equation}
We use $\tilde{\MC{C}}^{(E)}$ and $\tilde{\MC{C}}^{(V)}$ to denote the updated version of the constraints, i.e.,~(\ref{cons-C-E-updated}) and~(\ref{cons-C-V-updated}), in $\MC{C}^{(E)}$ and $\MC{C}^{(V)}$, respectively.
%
%Notice that both Constraint~(\ref{cons-C-E-updated}) and Constraint~(\ref{cons-C-V-updated}) depend only on $X$ and $H^*$.
%

\medskip

Let $\tilde{M}$ denote the submatrix of $\vphantom{{\dfrac{}{}}_{T}} M'_{(X,H)}$ formed by the rows belonging to $\tilde{\MC{C}}^{(E)} \cup \tilde{\MC{C}}^{(V)}$ and the columns belonging to $X \cup H^*$.
Lemma~\ref{claim-full-column-rank} follows from the extremality of $p$ and the fact that the constraints in $\vphantom{{\dfrac{}{}}^{\frac{T}{T}}} \MC{C}^{(E \times V)}$ are linearly independent.

\begin{lemma} \label{claim-full-column-rank}
$\tilde{M}$ has full column rank.
\end{lemma}

\begin{proof}%[Proof of Claim~\ref{claim-full-column-rank}]
By Proposition~\ref{claim-MID-full-column-rank}, the column vectors of the matrix $M_{(Eqs)}$ are linearly independent.
Therefore, the column vectors in the submatrix $M_{(X,H)}$ are also linearly independent as they do not involve in constraints other than those in $\vphantom{{\frac{}{}}^{\frac{T}{T}}} \OP{Eqs}(X,H)$.

\smallskip

Since $M_{(X,H)}$ contains an identity matrix, i.e., $M^{(1)}_{\MC{C}^{(E \times V)}}$, by Gaussian elimination we know that, taking out the rows and the columns that belong to $\vphantom{{\dfrac{}{}}^{\frac{T}{T}}_{T}} M^{(1)}_{\MC{C}^{(E \times V)}}$ from $M_{(X,H)}$ would have decreased its rank by exactly the dimension of $M^{(1)}_{\MC{C}^{(E \times V)}}$, i.e., $\BIGC{\vphantom{\frac{}{}} \hspace{2pt} \MC{C}^{(E \times V)} \hspace{2pt}}$.

\[
	\newcommand*{\tn}[3]{\tikz[remember picture]\node[#1] (#2) {#3};}
	\setlength{\arraycolsep}{0pt}
	\setlength{\dashlinedash}{2pt}
	\setlength{\dashlinegap}{1pt}
	\newcolumntype{C}[1]{>{\centering\arraybackslash}m{#1}}
	\newcolumntype{\Empty}{@{}m{0pt}@{}}
	\begin{array}{C{5em}cC{4em}c}
	$M'_{(X,H)} \colon$ &
		\left(
		\begin{array}{ C{4.2em} : C{4em} \Empty}
			$\mathbf{0}$ & $\tilde{M}$ & \\[1em] 
			\cdashline{1-2}
			$M^{(1)}_{\MC{C}^{(E\times V)}}$\parbox{1pt}{\tn{}{M1}{}} & 
			\tn{anchor=north,inner sep=1pt}{M1R}{}\hspace{2pt}$\cdots$\hspace{8pt} & \\[1.2em]
		\end{array}
		\right) &
		$\Longrightarrow$ &
		\left(
		\begin{array}{ C{4.2em} : C{4em} \Empty}
			$\mathbf{0}$ & $\tilde{M}$ & \\[1em] \cdashline{1-2}
			$M^{(1)}_{\MC{C}^{(E\times V)}}$ & $\mathbf{0}$ & \\[1.2em]
		\end{array}
		\right)
	\end{array}
	\tikzstyle{cvd_edge} = [thick,black!80,dashed,->]
	\tikz[remember picture,overlay]
	\path[->] (M1.east)[xshift=-0.4cm,yshift=-0.4cm] edge[cvd_edge,out=320,in=240] (M1R.south west);
\]

\smallskip

%
%
%This leaves a rank of $\BIGC{\vphantom{\frac{}{}} \hspace{2pt} \MC{C}^{(E)} \hspace{2pt}} + \BIGC{\vphantom{\frac{}{}} \hspace{2pt} \MC{C}^{(V)} \hspace{2pt}}$ for the matrix $\tilde{M}$.
%
%Since $\tilde{M}$ has exactly $\BIGC{\vphantom{\frac{}{}} \hspace{2pt} \MC{C}^{(E)} \hspace{2pt}} + \BIGC{\vphantom{\frac{}{}} \hspace{2pt} \MC{C}^{(V)} \hspace{2pt}}$ columns, it 
%\noindent
Therefore $\tilde{M}$ has a full column rank as well.
\end{proof}

%

%

%
%\paragraph{Mapping obtained from extremality.}
\subsection{Mapping obtained from extremality of $p$}
\label{sec-separation-mapping-from-extremality}
In the following we consider the matrix $\tilde{M}$.
Since $\tilde{M}$ has full column rank, for each column $j$ of $\tilde{M}$ there exists a distinct row $i$ such that $\tilde{M}(i,j)$ is non-zero\footnote{If not, the column reduction (Gaussian elimination on the columns) would have led to a rank less than $\BIGC{X \cup H^*}$, a contradiction to the fact that $\tilde{M}$ has a full column rank.}, i.e.,  the existence of distinct pivot for each column.
Let $$\sigma \colon X \cup H^* \mapsto \tilde{\MC{C}}^{(E)} \cup \tilde{\MC{C}}^{(V)}$$ denote one of such mappings.
Note that $\sigma$ is injective, i.e., for all $r, s \in X \cup H^*$, $r \neq s$ implies that $\sigma(r) \neq \sigma(s)$.
%
%
%\smallskip
%
%
For each $x_v \in X$, define a mapping $\pi \colon X \mapsto X \cup H^*$ as follows. 
$$
\pi(x_v) := \begin{cases}
h_{e,v}, \enskip \text{for \emph{any} $e$ such that } (e,v) \in H^*, \text{\enskip} & \text{if $v \in \MC{I}$ and $\sigma(x_v) = v$,} \\
x_v, & \text{otherwise.}
\end{cases}
$$
The following lemma shows that $\pi$ is a well-defined function of mapping.

\begin{lemma}
For any $v \in \MC{I}$ with $\sigma(x_{v}) = v$, there exists $e$ such that $(e,v) \in H^*$.
\end{lemma}

\begin{proof}
%
%\MC{C}^{(E \times V)}$.
%
Let $v \in \MC{I}$ be a vertex with $\sigma(x_v) = v$.
Since the image of $\sigma$ is $\tilde{\MC{C}}^{(E)} \cup \tilde{\MC{C}}^{(V)}$, from the assumption that $\sigma(x_v) = v$, we know that $v \in \tilde{\MC{C}}^{(V)}$.

\smallskip

Since $v \in \MC{I}$, we know that the variables $h_{e,v}$ for all $e \in E[v]$ with $h_{e,v} > 0$ have been included in $H$ by its definition.
Therefore the Constraint~(\ref{cons-C-V-updated}) for $v$ can be written as 
\begin{equation}
\sum_{\substack{e \text{ such that} \\ (e,v) \in H \setminus \MC{C}^{(E \times V)}}}d_e \cdot h_{e,v} \hspace{2pt} = \hspace{2pt} c'_v \cdot x_v.
\label{eq-pi-well-defined-degenerate-case}
\end{equation}
If $(e,v) \in \MC{C}^{(E \times V)}$ for all $e$ such that $(e,v) \in H$, then the left hand side of (\ref{eq-pi-well-defined-degenerate-case}) vanishes and we have $c'_v\cdot x_v = 0$.
Since $v \in \MC{I}$, $v$ is non-extremal and it follows that $x_v > \ell_v \ge 0$.
Therefore $c'_v$ must be zero.
This means that, the row to which $v$ corresponds in $\tilde{M}$ is a zero vector, which renders it impossible to be mapped to by $\sigma$ since it has no pivots for the column $x_v$, a contradiction.
Hence, there must exist $e$ such that $(e,v) \in H$ and $(e,v) \notin \MC{C}^{(E \times V)}$, which in turn implies that $(e,v) \in H^*$.
% and the value of $\pi(x_v)$ is well-defined.
%
\end{proof}

Consider the mapping $(\sigma \circ \pi) \colon X \mapsto \tilde{\MC{C}}^{(E)} \cup \tilde{\MC{C}}^{(V)}$. 
Since both $\sigma$ and $\pi$ are injective, 
%we know that 
$(\sigma \circ \pi)$ is also injective.
Lemma~\ref{claim-mapping-to-edges} provides an exact classification on the image of $(\sigma \circ \pi)$.

\begin{lemma} \label{claim-mapping-to-edges}
$(\sigma \circ \pi) (x_v) \in \tilde{\MC{C}}^{(E)}$ for all $v \in \MC{I}$. 
In contrast, $(\sigma \circ \pi)(x_v) = v$ for all $v \in \MC{D}$.
%\mong{polish the statement **}
\end{lemma}

\begin{proof}%[Proof of Llaim~\ref{claim-mapping-to-edges}]
First we show that $(\sigma \circ \pi) (x_v) \in \tilde{\MC{C}}^{(E)}$ if $v \in \MC{I}$.
Depending on whether or not $\sigma(x_v) = v$, we consider two cases.
If $\sigma(x_v) = v$, then it suffices to show that $\sigma(h_{e,v}) = e$ for all $e$ such that $(e,v) \in H^*$.
Indeed, for any $e$ such that $(e,v) \in H^*$, provided that $v \in \tilde{\MC{C}}^{(V)}$ (since it is mapped by $\sigma$), the variable $h_{e,v}$ appears in exactly two constraints in $\tilde{\MC{C}}^{(E)} \cup \tilde{\MC{C}}^{(V)}$, i.e., the constraint $e$ and the constraint $v$, respectively.
Since the constraint $v$ is already occupied by $x_v$ (in the mapping $\sigma$), it leaves $e$ the only choice for $h_{e,v}$ to occupy.
% in the mapping $\sigma$.
%
See also Figure~\ref{fig-mapping-on-edge}(a) for an illustration.

\smallskip

If $\sigma(x_v) \neq v$, then $(\sigma \circ \pi)(x_v) = \sigma(x_v)$ by definition. Since the variable $x_v$ involves in exactly one constraint in $\tilde{\MC{C}}^{(V)}$, i.e., the constraint $v$ itself, $\sigma(x_v) \neq v$ implies that $\sigma(x_v)$ must be in $\tilde{\MC{C}}^{(E)}$.
% if $\sigma(x_v) \neq v$.
%
This proves the first part of the statement.

\smallskip

To see that $(\sigma \circ \pi)(x_v) = v$ for all $v \in \MC{D}$, it suffices to see that in the matrix $M_{(X,H)}$, and hence the matrix $\tilde{M}$, the column belonging to $x_v$ has one and only one non-zero element that is located in the row belonging to the constraint $v \in \tilde{\MC{C}}^{(V)}$, i.e., in the diagonal matrix $M^{\MC{D}}_{\MC{C}^{(V)}}$.
Therefore we have $(\sigma \circ \pi)(x_v) = \sigma(x_v) = v$.
\end{proof}

\begin{figure*}[tp]
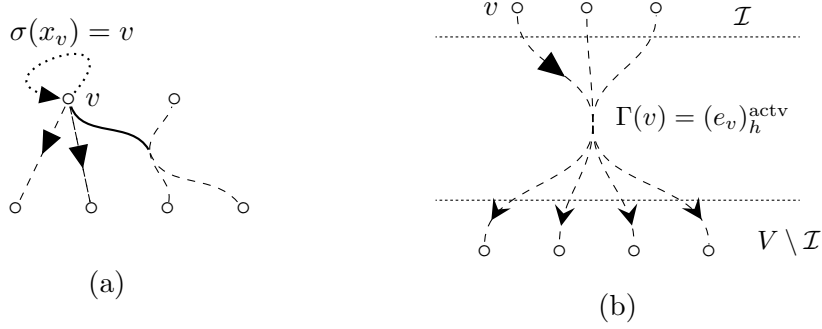

\centering
%
%\fbox
\hspace{2.2em}
{\begin{minipage}{10em}
\centering
\newcolumntype{C}[1]{>{\centering\arraybackslash}m{#1}}
\newcommand*{\tn}[3]{\tikz[remember picture]\node[#1] (#2) {#3};}
\setlength{\dashlinegap}{1pt}
\setlength{\dashlinedash}{1pt}
\begin{tabular}{@{} C{8em} @{}}
	\\[1em]
	\quad
	\begin{tabular}{ @{} *{2}{C{2.6em}} @{} }
	\quad \tn{anchor=base,inner sep=0pt}{v}{$\circ$} $v$ &
	\tn{anchor=base,inner sep=0pt}{v2}{$\circ$}
	\end{tabular}
	\\[4pt]
	\qquad \quad \tn{anchor=base,inner sep=0pt}{c}{} 
	\\[10pt]
	\begin{tabular}{ @{} *{2}{C{1.5em}} *{2}{C{1.5em}} @{} }
		\tn{anchor=base,inner sep=0pt}{u1}{$\circ$} &
		\tn{anchor=base,inner sep=0pt}{u4}{$\circ$} & 
		\tn{anchor=base,inner sep=0pt}{u2}{$\circ$} & 
		\tn{anchor=base,inner sep=0pt}{u3}{$\circ$}
	\end{tabular}
	\\[0.5em]
\end{tabular}
\tikzset{ptr/.style={
	decoration={markings,
		mark=at position 0.7 with 
			{\arrow[scale=1.8,rotate=0,>=latex']{#1}},
	},
	postaction={decorate}}
}
\newcommand*{\tkc}[1]{\tikz[remember picture,overlay] #1}
\newcommand*{\tkdarrow}[3]{\tikz[remember picture,overlay]
\draw[decoration={markings,mark=at position #3 with {\arrow[scale=1.4,>=triangle 45]{>}}},postaction={decorate},dashed] (#1) to (#2);
}
\newcommand*{\inptr}[4]{\tkoverlayarrow{0.6}{<}{scale=1.6,>=triangle 45}{#1}{#2}{#3}{#4}}
\newcommand*{\dptr}[1]{\tkoverlayarrow{0.85}{>}{scale=2.4,>=stealth}{hebase.south}{#1.north}{out=270,in=90}{dashed}}
%
%\inptr{hebase.north}{he1.south}{out=90,in=270}{dashed}
%\tkc{ \draw[-,dashed] (u1) to (v); }
\tkdarrow{v}{u1}{0.5}
\tkdarrow{v}{u4}{0.65}
\tkc{ \draw[-,dashed] (u4) to (v); }
\tkc{ \draw[-,dashed] (u2) to [out=80,in=280] (c); }
\tkc{ \draw[-,dashed] (u3) to [out=130,in=290] (c); }
\tkc{ \draw[-,thick] (c) to [out=120,in=290] (v); }
\tkc{ \draw[-,dashed] (c) to [out=90,in=250] (v2); }
\tkc{ \path[->,>=triangle 45] (v) edge[dotted,thick,loop left,min distance=14mm,in=170,out=50,looseness=14] node[anchor=south east,xshift=1.2cm] {$\sigma(x_v) = v$} (v); }
%\tkc{ \node (up) [anchor=base,inner sep=0pt,xshift=-0.5cm,yshift=-0.5cm] at (u4) {$\pi(h_{e,v}) = e$}; }
%
%\tikz[remember picture,overlay] \node at (hed.east) {\small $V \setminus \MC{I}$};
\tkc{ \node () [anchor=base,xshift=0.2cm,yshift=-1.1cm] at (u4) {(a)}; }
\end{minipage}}
%
%%%%%%%
%%%%%%%
%
\qquad \qquad
%
%%%%%%%%
%%%%%%%%
%
%\fbox
{\begin{minipage}{14em}
\centering
\newcolumntype{C}[1]{>{\centering\arraybackslash}m{#1}}
\newcommand*{\tn}[3]{\tikz[remember picture]\node[#1] (#2) {#3};}
\setlength{\dashlinegap}{1pt}
\setlength{\dashlinedash}{1pt}
\begin{tabular}{@{} c @{}}
	\begin{tabular}{ @{}c@{} *{3}{C{14pt}} @{} }
		$v$ &
		\tn{anchor=base,inner sep=0pt}{he1}{$\circ$} & 
		\tn{anchor=base,inner sep=0pt}{he2}{$\circ$} & 
		\tn{anchor=base,inner sep=0pt}{he3}{$\circ$}
	\end{tabular}
	\enskip \tn{anchor=east,inner sep=8pt}{heu}{}
	\\[4pt] \hdashline
	\\[10pt]
	\phantom{\qquad\qquad\quad} \tn{anchor=base,inner sep=0pt}{hebase}{} {\small \phantom{e}$\Gamma(v) = (e_v)^{\OP{actv}}_{h}$}
	\\[24pt] \hdashline
	\\[-4pt]
	\begin{tabular}{ *{4}{C{16pt}} }
		\tn{anchor=base,inner sep=0pt}{hed1}{$\circ$} & 
		\tn{anchor=base,inner sep=0pt}{hed2}{$\circ$} & 
		\tn{anchor=base,inner sep=0pt}{hed3}{$\circ$} &
		\tn{anchor=base,inner sep=0pt}{hed4}{$\circ$}
	\end{tabular}
	\tn{anchor=east,inner sep=6pt}{hed}{}
\end{tabular}
\tikzset{ptr/.style={
	decoration={markings,
		mark=at position 0.7 with 
			{\arrow[scale=1.8,rotate=0,>=latex']{#1}},
	},
	postaction={decorate}}
}
\newcommand*{\tkoverlayarrow}[7]{
\tikz[remember picture,overlay]
\draw[decoration={markings,mark=at position #1 with {\arrow[#3]{#2}}},postaction={decorate},#7] (#4) to [#6] (#5);
}
\newcommand*{\inptr}[4]{\tkoverlayarrow{0.6}{<}{scale=1.6,>=triangle 45}{#1}{#2}{#3}{#4}}
\newcommand*{\uline}[2]{\tikz[remember picture,overlay] \draw[-,#2] (hebase.north) to [out=90,in=270] (#1.south);}
\newcommand*{\dptr}[1]{\tkoverlayarrow{0.85}{>}{scale=2.4,>=stealth}{hebase.south}{#1.north}{out=270,in=90}{dashed}}
\inptr{hebase.north}{he1.south}{out=90,in=270}{dashed}
\uline{he2}{dashed} \uline{he3}{dashed}
\dptr{hed1} \dptr{hed2} \dptr{hed3} \dptr{hed4}
\tikz[remember picture,overlay] \node at (heu.south east) {\small $\MC{I}$};
\tikz[remember picture,overlay] \node at (hed.east) {\small $V \setminus \MC{I}$};
\smallskip

(b)
%\hspace{-80pt}
\end{minipage}}
\caption{The restricted behavior of the mapping $\sigma$. The arrow head shows the target to which $\sigma$ is mapping. (a) If the variable $x_v$ is mapped to $v$ itself, then there exists non-empty outward mapping to its incident edges. (b) When an edge constraint is occupied (mapped) by the variable $\pi(x_v)$, the remaining variables in the edge constraint 
have to map outward.}
\label{fig-mapping-on-edge}
\end{figure*}

Lemma~\ref{claim-mapping-to-edges} shows that the mapping $(\sigma \circ \pi)$ gives a distinct edge for each $v \in \MC{I}$.
The following lemma provides properties for the ``active ends'' of these particular edges.
Recall that $e^{\OP{actv}}_h$ denotes the set of the active ends of the edge $e$.

\begin{lemma}[\emph{Restricted Behavior of the Mappings}] \label{lemma-properties-active-ends-mapped-edges}
Consider the edge $e_v := (\sigma \circ \pi)(x_v)$ defined for any $v \in \MC{I}$.
We have 
\begin{enumerate}
	\item
		%$v \in \big( (\sigma \circ \pi)(x_v) \big)^{\OP{actv}}_h$ for all $v \in \MC{I}$.
		$v \in \big( e_v \big)^{\OP{actv}}_h$ for all $v \in \MC{I}$.

	\item
		%Furthermore, $\sigma(h_{(\sigma \circ \pi)(x_v),u}) = u$ for any $u \in \big( (\sigma \circ \pi)(x_v) \big)^{\OP{actv}}_h \setminus \MC{I}$.
		$\sigma(h_{e_v,u}) = u$ for any $u \in \big( e_v \big)^{\OP{actv}}_h \setminus \MC{I}$.
\end{enumerate}
\end{lemma}

\begin{proof}%[Proof of Lemma~\ref{lemma-properties-active-ends-mapped-edges}]
Consider a vertex $v \in \MC{I}$ and the edge $e_v := (\sigma \circ \pi)(x_v)$.
From the Constraint~(\ref{cons-C-E-updated}) for $e_v$ we know that $e_v$ is mapped either by the variable $h_{e_v,v}$ with $(e_v,v) \in H^*$ or the variable $x_v$, which in turn means that $(e_v,v) \in \MC{C}^{(E \times V)}$.
Therefore $(e_v,v) \in H$ in both cases.
It follows that $h_{e_v,v} > 0$ and $v \in \big(e_v\big)^{\OP{actv}}_h$.

\medskip

Below we prove the second part of this lemma.
Since the constraint to which $e_v$ corresponds is already occupied by $\pi(x_v)$ (in the mapping $\sigma$), it follows that, the remaining variables of $X \cup H$ that appear in constraint $e_v$ must be mapped to constraints other than $e_v$.
In particular, this statement holds for any $u \in \big( e_v \big)^{\OP{actv}}_h \setminus \MC{I}$.
If such a vertex $u$ exists, the variable $h_{e_v,u}$ would have been contained in $H$ since $\tilde{\MC{C}}^{(E)} = E[\MC{I}]$ and $u \in \big( e_v \big)^{\OP{actv}}_h$.
Since $\MC{I}$ is not supporting, it follows that $h_{e_v,u}$ involves in at most two constraints, i.e., $e_v$ and $u$.
Since $e_v$ is already mapped by $\pi(x_v)$, it follows that $u$ must be included in $\tilde{\MC{C}}^{(V)}$ for $h_{e_v,u}$ to be mapped to.
This shows that $\sigma(h_{e_v,u}) = u$.
Figure~\ref{fig-mapping-on-edge}(b) illustrates this argument.
\end{proof}

\newpage

%
%\paragraph{The mapping $\Gamma$.}
\subsection{The mapping $\Gamma$}
\label{sec-separation-mapping-gamma}
For each $v \in \MC{I}$, define $\Gamma(v) := \big( (\sigma \circ \pi)(x_v) \big)^{\OP{actv}}_h$.
Below we show that $\Gamma$ certifies the statement of Lemma~\ref{thm-structural-mapping}.

\begin{enumerate}%[(a)]
	\item
		(Reflexive) 
		
		This follows directly from Lemma~\ref{lemma-properties-active-ends-mapped-edges} and the definition of $\Gamma$.
		%, we have $v \in \big( (\sigma \circ \pi)(x_v) \big)^{\OP{actv}}_h$, which is exactly $\Gamma(v)$, for all $v \in \MC{I}$.
		
	\item
		(The image)% of $\Gamma$.)

		Consider a vertex $v \in \MC{I}$.
		Since $v \in \big( (\sigma \circ \pi)(x_v) \big)^{\OP{actv}}_h = \Gamma(v)$ by the reflexive property, it follows that 
		\begin{equation}
		\Gamma(v) \in E^{\OP{actv}}_h[v] \subseteq E^{\OP{actv}}_h[\MC{I}]. \label{lemma-separation-mapping-gamma-image-1}
		\end{equation}

		For any $u \in \MC{D}$, from Lemma~\ref{claim-mapping-to-edges} we know that $\sigma(x_u) = u$.
		Therefore, by Lemma~\ref{lemma-properties-active-ends-mapped-edges}, $u$ cannot belong to $\Gamma(v)$ since $u \notin \MC{I}$ and since the constraint $u$ is already occupied by $x_u$ in the mapping $\sigma$.
		This implies that $\Gamma(v) \notin E^{\OP{actv}}_h[u]$.
		Since this holds for all $u \in \MC{D}$, we have $\Gamma(v) \notin E^{\OP{actv}}_h[\MC{D}]$.
		Combined with (\ref{lemma-separation-mapping-gamma-image-1}), it follows that $$\Gamma(v) \in E^{\OP{actv}}_h[\MC{I}] \setminus E^{\OP{actv}}_h[\MC{D}].$$
		
	\item
		(Closed under intersection.)
		
		Consider any $u, v \in \MC{I}$ such that $u \neq v$.
		For any $w \in \Gamma(u) \cap \Gamma(v)$, if $w \notin \MC{I}$, then it follows from Lemma~\ref{lemma-properties-active-ends-mapped-edges} that $\sigma(h_{(\sigma \circ \pi)(x_u),w}) = w$ and $\sigma(h_{(\sigma \circ \pi)(x_v),w}) = w$, a contradiction to the injective property of $\sigma$.
		Therefore $\Gamma(u) \cap \Gamma(v) \subseteq \MC{I}$.
\end{enumerate}

This proves Lemma~\ref{thm-structural-mapping}.

\end{appendix}

\end{document}